\documentclass[twocolumn,floatfix,superscriptaddress,pra,aps,10pt,longbibliography]{revtex4-2}
\usepackage{times}

\usepackage[normalem]{ulem}
\usepackage{amssymb}
\usepackage{color}
\usepackage{amsmath}
\usepackage{amsbsy}
\usepackage{amsthm}
\usepackage{graphicx}
\usepackage{bbm}
\usepackage{bm}
\usepackage{epsfig}
\usepackage{xfrac}
\usepackage{xcolor}
\usepackage{enumerate}
\usepackage{multirow}
\usepackage{orcidlink}
\definecolor{myurlcolor}{rgb}{0,0,0.7}
\definecolor{myrefcolor}{rgb}{0.8,0,0}
\hypersetup{
unicode=true,pdfusetitle, bookmarks=true,bookmarksnumbered=false,bookmarksopen=false, breaklinks=false,pdfborder={0 0 0},backref=false,colorlinks=true, linkcolor=myrefcolor,citecolor=myurlcolor,urlcolor=myurlcolor}

\newcommand{\ket}[1]{\left| {#1} \right\rangle}
\newcommand{\bra}[1]{\left\langle {#1}\right|}

\newcommand{\braket}[2]{\langle #1|#2\rangle}

\renewcommand{\t}[1]{\textrm{#1}}
\newcommand{\tr}[0]{\mathrm{Tr}}

\newcommand{\vspan}{\mathrm{span}}

\usepackage{braket}
\newcommand{\thmref}[1]{\hyperref[#1]{Theorem~\ref{#1}}}
\newcommand{\lemmaref}[1]{\hyperref[#1]{Lemma~\ref{#1}}}
\newcommand{\figref}[1]{\hyperref[#1]{Fig.~\ref{#1}}}
\newcommand{\figaref}[1]{\hyperref[#1]{Fig.~\ref{#1}a}}
\newcommand{\figbref}[1]{\hyperref[#1]{Fig.~\ref{#1}b}}
\newcommand{\figcref}[1]{\hyperref[#1]{Fig.~\ref{#1}c}}
\renewcommand{\eqref}[1]{\hyperref[#1]{Eq.~(\ref{#1})}}
\newcommand{\eqsref}[2]{\hyperref[#1]{Eqs.~(\ref{#1})-(\ref{#2})}}
\newcommand{\appref}[1]{\hyperref[#1]{Appx.~\ref{#1}}}

\usepackage{dsfont}

\newtheorem{theorem}{Theorem}
\newtheorem{lemma}{Lemma}

\usepackage{tabularray}
\DeclareFontFamily{U}{zavm}{}
\DeclareFontShape{U}{zavm}{m}{n}{<-> favmr7y}{}
\DeclareRobustCommand{\cmark}{{\usefont{U}{zavm}{m}{n}\char128}}
\DeclareRobustCommand{\xmark}{{\usefont{U}{zavm}{m}{n}\char129}}

\begin{document}

\title{Protecting Heisenberg scaling in quantum metrology via engineered dressed states}

\author{Wojciech G{\'o}recki\,
\orcidlink{0000-0001-9912-9186}}
\email{wojciechgorecki.qi@gmail.com}
\affiliation{Freie Universit{\"a}t Berlin, Fachbereich Physik and Dahlem
Center for Complex Quantum Systems, Arnimallee 14, 14195 Berlin, Germany}
\author{Christiane P. Koch\,
\orcidlink{0000-0001-6285-5766}}
\email{christiane.koch@fu-berlin.de}
\affiliation{Freie Universit{\"a}t Berlin, Fachbereich Physik and Dahlem
Center for Complex Quantum Systems, Arnimallee 14, 14195 Berlin, Germany}

\begin{abstract}
Quantum metrology promises precision beyond classical limits but environmental noise, unless properly controlled, reduces the quantum advantage to at most a constant improvement. A key challenge is therefore to design quantum control strategies that suppress noise while preserving sensitivity to the targeted signal. Here, we suggest to use dressed states generated by 
a constant Hamiltonian to achieve this goal and show that success of this strategy depends on the spectral properties of the environment. For low-temperature noise, we show that Heisenberg scaling can be achieved if and only if the signal generator lies outside the linear span of the system–environment coupling operators. This implies that the proper dressed states may enable Heisenberg scaling even in cases where the well-known Hamiltonian-not-in-Lindblad-span criterion, evaluated without dressing, would forbid it. We illustrate dressed state metrology
for the example of
 NV-center thermometry under magnetic-field fluctuations, with the framework readily applicable to other platforms as well.
\end{abstract}

\maketitle

{\it Introduction---}Quantum metrology aims to estimate parameters encoded in the Hamiltonian of a system by preparing a probe state, letting it evolve (possibly under external control), performing a measurement, and inferring the parameter from the outcomes~\cite{degen2017quantum}. Optimizing such protocols~\cite{ma2021adaptive,marciniak2022optimal,liu2024fully,hernandez2024optimal,kurdzialek2025quantum}
can significantly enhance precision over standard methods~\cite{giovannetti2011advances,toth2014quantum,liu2017quantum,huang2024entanglement}, with applications ranging from quantum clocks \cite{bloom2014optical,ludlow2015optical,nicholson2015systematic,campbell2017fermi,zhang2024frequency}, gravitometry \cite{bothwell2022resolving,bongs2019taking,rosi2014precision,kasevich1991atomic}, and gravitational-wave detection \cite{caves1981quantum,tse2019quantum,acernese2019increasing}, to magneto- and electrometry \cite{kucsko2013nanometre,dolde2011electric,barry2020sensitivity,zhou2020quantum}, quantum imaging \cite{brida2010experimental,tsang2016quantum,casacio2021quantum}, and accelerometry \cite{cheiney2018navigation}. 
A key ingredient behind quantum enhancement is the coherent accumulation of signal: under noiseless unitary evolution, it can grow both with the interrogation time $T$ and the number of probes $N$ (via entanglement), yielding the mean squared error with Heisenberg scaling (HS) $\sim 1/(NT)^2$ ~\cite{giovannetti2006quantum,Paris2009,giovannetti2011advances,demkowicz2015optical,Schnabel2016,degen2017quantum,Pezze2018,Pirandola2018,wang2019heisenberg,huang2024entanglement,puig2025dynamical}. In the presence of generic noise, however, decoherence and dissipation suppress both entanglement and long-time signal accumulation, leading to the standard scaling $\sim 1/(NT)$~\cite{fujiwara2008fibre,escher2011general,demkowicz2012elusive,kolodynski2013efficient,knysh2014true,demkowicz2014using,sekatski2017quantum,layden2018ancilla,zhou2021asymptotic,kurdzialek2022using,das2025universal}. In certain cases, suitably designed controls can restore unitary dynamics, for example through noise filtering~\cite{paz2014general}, dynamical decoupling~\cite{viola1999dynamical,green2013arbitrary}, continuous dynamical decoupling~\cite{chen2006geometric,soare2014experimental}, dressed states induced by micro- or radio-frequency driving~\cite{timoney2011quantum,xu2012coherence,golter2014protecting,hirose2012continuous} or static external fields~\cite{dolde2011electric}, as well as quantum error correction (QEC)~\cite{arrad2014increasing,unden2016quantum,chen2024quantum}. However, in a metrological setting such techniques may also suppress the signal itself. It is therefore important to identify in which cases noise can be mitigated without losing sensitivity to the signal, thus allowing HS.

Solutions to this problem 
have been found in the context of noise-resilient entangled states~\cite{escher2011general,demkowicz2012elusive,fujiwara2008fibre,kolodynski2013efficient,knysh2014true} and quantum error correction for Markovian noise~\cite{demkowicz2014using,zhou2021asymptotic,sekatski2017quantum,kurdzialek2022using,layden2018ancilla,das2025universal}: The Hamiltonian-not-in-Lindblad-span (HNLS) criterion for uncorrelated noise identifies when quantum error correction can recover HS for systems governed by a master equation~\cite{demkowicz2017adaptive,zhou2018achieving,wan2022bounds}. However, 
one could also attempt to prevent noise accumulation during the evolution, rather than correcting errors after they occur.
For non-Markovian dynamics,
temporally correlated noise can be mitigated more effectively than uncorrelated noise~\cite{altherr2021quantum,kurdzialek2025universal,mann2025quantum}. Specifically, for qubits under dephasing, working on timescales shorter than the bath correlation time improves the precision beyond standard scaling, but not up to HS~\cite{chin2012quantum,macieszczak2015zeno,smirne2016ultimate,haase2018fundamental,riberi2022frequency,riberi2026precision}; more generally, HS can be achieved with qubit sensors only for particular types of noise
~\cite{sekatski2016dynamical} or nonlinear encoding~\cite{riberi2025optimal}. Since
extra dimensions are known to be a useful control resource~\cite{basilewitsch2021fundamental,ticozzi2014quantum,ticozzi2017quantum}, it is natural to ask whether HS can be achieved under general noise in higher dimensions.

In this work, we propose to achieve HS by preventing noise accumulation using dressed states in sensors with more than two levels: for a given system--environment interaction Hamiltonian, the choice of dressed states modifies the effective noise acting on the system~\cite{beaudoin2011dissipation}, as readily seen in e.g.\ adiabatic master equations~\cite{albash2012quantum}.
Assuming access to a noiseless ancilla, we derive necessary and sufficient conditions for HS under different bath spectra; for low-temperature noise, which results in dephasing and energy relaxation, we show that an appropriate choice of dressed states may enable HS even when the HNLS criterion~\cite{demkowicz2017adaptive,zhou2018achieving,wan2022bounds}, evaluated in the absence of dressing, would forbid it.
As an illustration, we consider as sensor a spin-1 particle, inspired by  thermometry experiments which successfully decoupled an NV center from a slowly varying isotropic magnetic field (the source of dephasing) without compromising the signal~\cite{yun2021temperature,tabuchi2023temperature,beaver2024selective}. Accounting additionally for energy relaxation, we find that 
entanglement with a noiseless auxiliary system is necessary to achieve HS.
Our analysis identifies where genuine room for improvement remains, providing a basis for including platform-specific constraints and more realistic noise models.

{\it Decoupling from classical noise---}Consider a system coupled to an environment (``bath''), with total
Hamiltonian
\begin{equation}
    H=H_S\otimes\openone_B+\openone_S\otimes H_B+H_I,
\end{equation}
where the interaction part reads
\begin{equation}
\label{eq:ab}
    H_I=\sum_\alpha A_\alpha\otimes B_\alpha,
\end{equation}
with Hermitian operators $A_\alpha$ acting on the system and $B_\alpha$ on the
environment.
The free Hamiltonian of the system depends on a parameter $\omega$, and our goal
is to estimate its small deviation from a known value $\omega_0$,
$\omega=\omega_0+\delta\omega$. To first order,
$H^{\t{free}}_S(\omega)\approx H^{\t{free}}_S(\omega_0)+\delta\omega\,G$, where
$G=dH^{\t{free}}_S/d\omega$ is the signal generator. We assume control over the
system through a static term $H_C$, so that the total system Hamiltonian reads
\begin{equation}
\label{eq:hsp}
    H_S=H^{\t{free}}_S(\omega_0)+\delta\omega\,G+H_C
      =H_D+\delta\omega\,G,
\end{equation}
where $H_D=H^{\t{free}}_S(\omega_0)+H_C$ and its eigenstates are referred to as dressed states. A static $H_C$ may be realized either by
applying a constant field to the system or effectively via a time-periodic (e.g., laser or microwave) drive.

To illustrate the decoupling mechanism,
consider first the simpler situation in which the system is weakly coupled
to classical scalar fields slowly fluctuating around zero,
\begin{equation}
    H_I\to H_{\t{cl}}(t)=\sum_\alpha \xi_\alpha(t)\,A_\alpha,\quad H=H_S+H_{\t{cl}}(t).
\end{equation}
Let $\Pi_\epsilon$ denote the
projection onto the eigenspace of $H_D$ with energy $\epsilon$. If the distinct eigenvalues of $H_D$ are sufficiently well
separated, the off-diagonal blocks of both $\delta\omega\,G$ and $H_{\t{cl}}(t)$
are suppressed by fast oscillations. To leading order, the evolution is then 
governed by the effective Hamiltonian
\begin{equation}
\label{eq:hamproj}
H_{\t{eff}}(t)=H_D+\delta\omega\sum_\epsilon\Pi_\epsilon G \Pi_\epsilon
  +\sum_{\epsilon,\alpha}\xi_\alpha(t)\Pi_\epsilon A_\alpha\Pi_\epsilon.
\end{equation}
Identifying $H_D$ as a strong driving whose spectral projections define the
Zeno subspaces, and the signal and noise terms
$V(t)=\delta\omega\,G+H_{\t{cl}}(t)$ as perturbations~\footnote{Note that reaching the
Zeno regime here requires only that the
\emph{total} dressing Hamiltonian be strong compared with the noise operators,
while $H_C$ itself may be of the same order of magnitude as
$H^{\t{free}}_S(\omega_0)$.}, we can invoke the
strong-driving bounds of quantum Zeno
dynamics~\cite{facchi2002quantum,Facchi_2008} to control the deviation of the
exact evolution from the one generated by
$H^{\t{eff}}$~\cite[Remark~14]{Burgarth2022oneboundtorulethem}:
\begin{multline}
\label{eq:timesscales}
\Big\|\mathcal{T}e^{-i\int_0^{T}\!dt\,H(t)}-\mathcal{T}e^{-i\int_0^{T}\!dt\,H_{\t{eff}}(t)}\Big\|
\le\\
\frac{\sqrt m}{\nu_{\min}}\big(1+2T\|V\|_{\infty,T}\big)
\big(2\|V\|_{\infty,T}+T\|\dot V\|_{\infty,T}\big),
\end{multline}
where $\|X\|_{\infty,T}=\sup_{t\in[0,T]}\|X(t)\|$, $\nu_{\min}$ is the smallest
gap between distinct eigenvalues of $H_D$, and $m$ is their number. The
projected description \eqref{eq:hamproj} is thus accurate provided the
perturbation is weak compared with the dressing gap,
$\sqrt m\,\|V\|_{\infty,T}\ll\nu_{\min}$, and the interrogation time is short
compared with the scales set by $\|V\|_{\infty,T}$ and $\|\dot V\|_{\infty,T}$;
see Supplemental Material (SM)~\ref{app:zeno}.

Since the $\xi_\alpha(t)$ are unknown, the residual couplings
$\Pi_\epsilon A_\alpha\Pi_\epsilon$ dephase the Zeno subspaces. 
Heisenberg
scaling therefore requires a subspace on which they act trivially, yet which
remains sensitive to the signal.
Assuming that the $\xi_\alpha(t)$ for different $\alpha$ are not
fully correlated, this requires two eigenstates of $H_D$, $\ket{\psi_0}$,
$\ket{\psi_1}$, with eigenenergies $\lambda_0\neq\lambda_1$, such that
\begin{equation}
\begin{split}
\label{eq:dephasing}
   \forall_\alpha\ \ &\braket{\psi_0|A_\alpha|\psi_0}
   =\braket{\psi_1|A_\alpha|\psi_1}\\
   \t{and}\quad &\braket{\psi_1|G|\psi_1}-\braket{\psi_0|G|\psi_0}>0,
   \end{split}
\end{equation}
so the effective signal generator $G_{\t{eff}}=\braket{\psi_0|G|\psi_0}\ket{\psi_0}\bra{\psi_0}+\braket{\psi_1|G|\psi_1}\ket{\psi_1}\bra{\psi_1}$ acts non-trivially on this subspace.
Once two such states are found, they can always be engineered as the lowest
eigenstates of $H_D$ by choosing
$H_C=-H^{\t{free}}_S(\omega_0)+\lambda_0\ket{\psi_0}\bra{\psi_0}
+\lambda_1\ket{\psi_1}\bra{\psi_1}$ with $\lambda_0<\lambda_1<0$. In
many practical cases much simpler control suffices, as we discuss below in
\textit{Application}. Designing the protocol requires knowledge of the coupling operators $A_\alpha$, but not of the noise statistics.

A necessary condition for the existence of such a pair of states is 
$G\notin\mathrm{span}_{\mathbb R}\{\openone,A_\alpha\}$. This becomes sufficient once the system is supplemented by a noiseless ancilla
under full control, i.e.\ for $G\otimes\openone$ and $A_\alpha\otimes\openone$
on $\mathcal H_S\otimes\mathcal H_A$, with $H_C$ acting jointly on both.

\begin{lemma}[Algebraic criterion]
\label{lem:core}
Assuming auxiliary system with $\dim\mathcal{H}_A\ge 2\dim\mathcal{H}_S$, a pair of orthogonal vectors
$|\psi_0\rangle,|\psi_1\rangle\in\mathcal{H}_S\otimes\mathcal{H}_A$ satisfying \eqref{eq:dephasing} exists if and only if
\begin{equation}
\label{eq:lemma}
   G\notin \mathrm{span}_{\mathbb R}\{\openone,A_\alpha\}.
\end{equation}
The maximization of
$\langle\psi_1|G|\psi_1\rangle-\langle\psi_0|G|\psi_0\rangle$ subject to these
constraints can be formulated as a semidefinite program. 
\end{lemma}
\begin{proof}
The proof adapts the method of~\cite{zhou2018achieving} to a different operator
space; see End Matter~\ref{app:proof}.
\end{proof}

Once a control Hamiltonian $H_C$ has been applied such that
$\ket{\psi_0},\ket{\psi_1}$ are energy-separated eigenstates of $H_D$, the
metrological protocol consists of preparing the input state
$\ket{\psi_{\t{in}}}=\frac{1}{\sqrt{2}}(\ket{\psi_0}+\ket{\psi_1})$, letting it evolve for a time $T$, and performing a projective measurement.
For a large number of repetitions $k$, the variance of the optimal estimator asymptotically saturates the quantum Cram\'er--Rao bound $\Delta^2\tilde{\omega}\approx1/(kF_Q)$~\cite{degen2017quantum}. Here, the quantum Fisher information per trial is given by $F_Q=4T^2\Delta^2 G_{\t{eff}}$, where $\Delta^2 G_{\t{eff}}$ is the variance of the effective signal generator in the input state. This yields Heisenberg scaling in $T$, for times at which the projected
description remains accurate.
More generally, with access to $N$ such subsystems one may prepare an entangled input state for which the variance of the total effective generator scales as $N^2$, giving $F_Q\sim(NT)^2$; maintaining the same window of approximately unitary evolution then requires a stronger dressing, scaling with $N$, see SM \ref{app:zeno}.

{\it Master equation in the weak-coupling limit---}Returning to the quantum bath of~\eqref{eq:ab}, we now analyze 
the potential decoupling from the noise 
within the master equation formalism, relating it to the bath spectral
density.
We assume the standard Born--Markov approximation, i.e., weak coupling, and a bath correlation time much shorter than the relaxation time of the
system, together with the rotating-wave (secular) approximation. Then the reduced dynamics is described by the
Gorini--Kossakowski--Sudarshan--Lindblad (GKSL) 
equation~\cite{breuer2002theory},
\begin{equation}
\label{eq:master}
    \frac{d}{dt}\rho_S(t)=-i[H_S+H_{LS},\rho_S(t)]+\mathcal D[\rho_S(t)]\,.
\end{equation}
Here $H_{LS}$ is the Lamb-shift Hamiltonian and $\mathcal D[\cdot]$ describes the dissipative part of the evolution,
\begin{equation}
\label{eq:d}
   \mathcal D[\rho_S]=\sum_\nu \sum_{\alpha,\beta}\gamma_{\alpha\beta}(\nu)
   \left(L_\beta^{\nu}\rho_S{L_\alpha^{\nu}}^\dagger
   -\frac{1}{2}\{{L_\alpha^{\nu}}^\dagger L_\beta^{\nu},\rho_S\}\right)\,,
\end{equation}
where $\gamma_{\alpha\beta}(\nu)$ is the bath spectral density and $L_\alpha^\nu$ the jump operator corresponding to the energy gap $\nu$,
\begin{equation}
\label{eq:jump}
  L_\alpha^\nu\equiv \sum_{\epsilon'-\epsilon=\nu}\Pi_{\epsilon}A_\alpha
  \Pi_{\epsilon'}.
\end{equation}
As above, the projections are taken onto the eigenspaces of $H_D$ since $H_S=H_D+\delta\omega\,G$ with $\delta\omega$
small. We assume throughout that $\gamma_{\alpha\beta}(\nu)$ has full rank for every
$\nu$ at which it does not vanish, i.e.\ the noises generated by different
$A_\alpha$ are neither fully correlated nor anticorrelated, 
analogously to the assumption on the $\xi_\alpha(t)$ above. This is the least favorable case, 
and a protocol that cancels the noise for full-rank
$\gamma_{\alpha\beta}(\nu)$ remains valid when the rank is deficient, whereas the
converse need not hold~\footnote{Moreover, whenever the spectral density factorizes, 
$\gamma_{\alpha\beta}(\nu)=f(\nu)\gamma_{\alpha\beta}$ with scalar $f(\nu)$, a
rotation in operator space diagonalizes $\gamma_{\alpha\beta}$, reducing a rank-deficient case to a full-rank one with a smaller set of $A'_\alpha$.}. Recall that $\nu=0$ corresponds to dephasing, $\nu>0$ to
relaxation (decay), and $\nu<0$ to thermal excitation. For a low-temperature
environment $\gamma_{\alpha\beta}(\nu)\approx0$ for all $\nu<0$, i.e.\ thermal
excitation is negligible~\footnote{If a static $H_C$ is realized by periodic driving, the drive may open
transitions at shifted frequencies. This is avoided for a far-detuned drive
adiabatically eliminating an auxiliary level~\cite{brion2007adiabatic}}.
Dephasing caused by a slowly fluctuating classical field is
recovered here when $\gamma_{\alpha\beta}(\nu)$ is peaked around $\nu=0$ with a width (roughly the inverse bath correlation time) much smaller than the energy gaps of $H_D$.

The dependence of the Lindblad operators on the eigenstates of $H_D$
opens the possibility of creating decoherence-free dressed subspaces---i.e., subspaces spanned by a subset of eigenstates of the modified system Hamiltonian that remain invariant under the non-unitary part of the evolution. Our goal is to identify when a control Hamiltonian can be found such that, for the resulting master equation, either a decoherence-free dressed subspace sensitive to variations of $\omega$ exists---thus immediately enabling HS---or HS can be recovered through quantum error correction.


{\it Conditions for achievability of HS---}
First we focus on low temperature noise. 
In the absence of thermal excitation,
the problem of maximizing the signal while preserving unitary evolution may be simplified to the following optimization problem:

\begin{subequations}
\label{eq:cond}
\begin{align}
\max_{\ket{\psi_0},\ket{\psi_1}}&\braket{\psi_1|G|\psi_1}-\braket{\psi_0|G|\psi_0}, \quad\t{with}\label{eq:maxG}\\
\t{(dephasing)}&\quad
   \forall_\alpha\ \ \braket{\psi_0|A_\alpha|\psi_0}
   =\braket{\psi_1|A_\alpha|\psi_1},
  \label{eq:dephasing2}\\
   \t{(relaxation)}&\quad\forall_\alpha\ \ \braket{\psi_0|A_\alpha|\psi_1}=0,
  \label{eq:relaxation}
\end{align}
\end{subequations}
where $\ket{\psi_0},\ket{\psi_1}$ are the two lowest eigenstates of $H_D$.
These conditions imply that all operators $A_\alpha$ act proportionally to the identity when projected onto subspace $\mathcal C := \vspan_{\mathbb C}\{\ket{\psi_0},\ket{\psi_1}\}$, i.e., $\Pi_\mathcal C A_\alpha \Pi_\mathcal C \propto \Pi_\mathcal C$.
The Lamb shift $H_{LS}$
commutes with $H_D$ by construction and, as long as its value is known, does
not affect the estimation procedure; see End Matter~\ref{app:lamb}.
The optimization~\eqref{eq:cond} may give a non-zero result within
$\mathcal H_S$ alone in particular cases, but only with an ancilla is it
guaranteed in general, by the following theorem.

\setcounter{theorem}{0}
\begin{theorem}
\label{thm:1}
For the master equation in the weak-coupling limit
in the presence of dephasing and relaxation noise and assuming full control over both the system and a noiseless ancilla, a decoherence-free subspace useful for sensing $\delta\omega$ can be generated by properly dressing the Hamiltonian if and only if
\begin{equation}
\label{eq:hncs}
   G\notin \mathrm{span}_{\mathbb R}\{\openone,A_\alpha\,\,\t{for all}\,\,\alpha\}.
\end{equation}
\end{theorem}
\begin{proof}
The theorem follows directly from Lemma \ref{lem:core}, since the additional condition 
\eqref{eq:relaxation} may always be satisfied by choosing $\ket{\psi_0}$ and $\ket{\psi_1}$ such that they have orthogonal support on $\mathcal H_A$.
\end{proof}

The decoupling mechanism relies on the noise retaining a finite correlation
time, over which the corresponding jumps average out to zero, even when the
evolution is Markovian on the timescales of interest. Whenever the bath
spectral density has nonzero support at all frequencies---whether because the
noise is white (see SM~\ref{app:perpendicular}) or because the bath is
thermal---excitation can no longer be avoided and a further constraint appears:
\begin{equation}
\label{eq:spexcond}
\t{(excitation)}\quad\forall_\alpha\forall_{i\neq 0,1}\braket{\psi_i|A_\alpha|\psi_{0/1}}=0,
\end{equation}
where $\ket{\psi_i},i\neq 0,1$ are the remaining eigenvectors of $H_D$.
Since a decoherence-free dressed subspace is then generally impossible, it is
natural to relax complete noise suppression and instead seek one in which errors
may occur but remain correctable by quantum error correction.

The possibility of achieving Heisenberg scaling using QEC for \textit{fixed} Lindblad operators has been fully characterized in terms of the Hamiltonian-not-in-Lindblad-span (HNLS) theorem \cite{demkowicz2017adaptive,zhou2018achieving}. It states that, for a system governed by a master equation and allowing for the most general metrological protocol—including noiseless ancillas and arbitrary quantum error correction applied during the evolution—Heisenberg scaling can be achieved if and only if $G\notin\vspan_{\mathbb C}\{\openone,L_i,L^\dagger_i, L_j^\dagger L_i\}$ (here the single index $i$ covers both $\nu,\alpha$ in \eqref{eq:d}). This theorem relies on the assumption that the applied control (including QEC) does not modify the Lindblad operators. The assumption is met when QEC  
is implemented via strong, effectively instantaneous control pulses separated by time intervals longer than the coarse-graining timescale used to derive the master equation.
We extend this theorem by considering constant control in the weak-coupling regime, where the effective Lindblad operators are determined by both the (uncontrollable) coupling operators $A_\alpha$ and the eigenbasis of the (controllable) dressing Hamiltonian.

\begin{theorem}
\label{thm:2}
For the master equation in the weak-coupling limit
for fully general noise (including thermal excitation) and assuming full control over both the system and the ancilla, a correctable subspace—i.e., one for which an effective unitary evolution can be obtained after properly applying quantum error correction—useful for sensing $\delta\omega$ can be generated by appropriately dressing the system if and only if
\begin{equation}
\label{eq:hncs2}
G\notin \mathrm{span}_{\mathbb C}\{\openone,A_\alpha,A_\alpha A_\beta\,\,\t{for all}\,\,\alpha,\beta\}.
\end{equation}
Namely, this condition is equivalent to the existence of a control Hamiltonian $H_C$ acting on $\mathcal H_S\otimes \mathcal H_A$, for which
$G\notin\vspan_{\mathbb C}\{\openone,L_i,L^\dagger_i, L_j^\dagger L_i\}$
with 
the modified Lindblad operators given by \eqref{eq:jump}.
\end{theorem}
The proof is given in End Matter \ref{app:proof2}. Note that the span is taken over $\mathbb{C}$ rather than $\mathbb{R}$ (the listed elements are not necessarily Hermitian and constructing a Hermitian $G$ may require complex coefficients).

{\it Application and discussion---}We demonstrate the usefulness of our theory by considering a spin-1 particle with system and interaction Hamiltonians:
\begin{equation}
H^{\t{free}}_S=(D+\delta \omega) S_z^2,\quad H_I=\gamma_e\vec{S}\cdot \vec{B},
\end{equation}
so $G=S_z^2$ and $\forall_{\alpha\in\{x,y,z\}}A_\alpha=S_\alpha$.
Without going into the details of the bath Hamiltonian, we will consider three different spectral bath functions. As a concrete realization we consider an NV center ~\cite{dolde2011electric,doherty2013nitrogen,rondin2014magnetometry}, where the parameter $\delta \omega$ corresponds to temperature changes~\cite{tabuchi2023temperature}. A realistic description of the NV center requires a non-Markovian noise model~\cite{gaebel2006room,hanson2008coherent}, but to keep the correspondence with HNLS transparent we work within the Markovian master-equation formalism. Still, the decoupling mechanism discussed below is not restricted to this setting; as the classical-field analysis above shows, it applies to temporally correlated noise as well.

We first consider the case, where the spectral bath function is focused around zero, so the only type of noise appearing in the system is dephasing. In the absence of control, noise arising from fluctuations in the $x$ and $y$ directions is suppressed, while fluctuations along $z$ dominate, i.e. $L\propto S_z$. The standard HNLS condition is not satisfied, $G=S_z^2\in\vspan_{\mathbb C}\{\openone,S_z,S_z^2\}$, so HS is not achievable by any protocol acting on timescales where the dynamics is described by the averaged master equation, regardless of access to a noiseless ancilla. In contrast, the condition from \autoref{thm:1},
$G=S_z^2\notin\vspan_{\mathbb R}\{\openone,S_x,S_y,S_z\}$, is satisfied. Here the ancilla is in fact not needed: 
the states maximizing the signal with condition $\forall_{\alpha\in\{x,y,z\}}\bra{\psi_{0}}S_{\alpha}\ket{\psi_{0}}=\bra{\psi_{1}}S_{\alpha}\ket{\psi_{1}}$ may be found already in $\mathcal H_S$: $\ket{\psi_0}=\ket{0}$ and
$\ket{\psi_1}=\ket{\psi_{\pm}}=\tfrac{1}{\sqrt{2}}(\ket{+1}\pm\ket{-1})$, with $\ket{-1},\ket{0},\ket{+1}$ the eigenstates of $S_z$, and the control $H_C\propto S_x^2-S_y^2$ makes them eigenstates of $H_D$. Decoupling by a static control was among the strategies employed in~\cite{tabuchi2023temperature}, with a perpendicular magnetic field approximating $H_C$.
Within the approximations taken, the input state
$\ket{\psi_{\t{in}}}=\tfrac{1}{\sqrt{2}}(\ket{0}+\ket{\psi_-})$ evolves
unitarily, which leads to HS.
The resilience of this eigenbasis against the noise is due
to the expectation value of any spin component being zero
for each of these vectors. Note that such a basis can exist
only for spin $\geq 1$ since any pure state of a spin-$1/2$ points in a definite direction. This explains why HS cannot be
achieved for a qubit under isotropic coupling, irrespective of the signal
generator~\cite{sekatski2016dynamical}.

\begin{figure}[t!]
  \includegraphics[width=0.46
\textwidth]{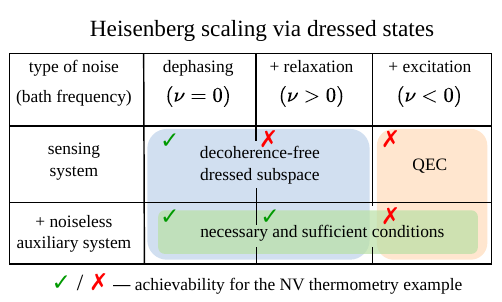}
\centering
\caption{
  Possibility of achieving Heisenberg scaling (HS) in quantum metrology using dressed states in the presence of three types of noise, corresponding to different bath spectra: dephasing,
  relaxation, and thermal excitation. In successive columns, additional noise processes are included, with all previously introduced noise sources 
  also present. For the first two cases, 
  decoherence-free dressed subspaces can be constructed, whereas the last case 
  requires the application of quantum error correction (QEC). When the sensing system is supplemented with a noiseless auxiliary system,
  necessary and sufficient conditions for achieving HS exist in all cases. Applying these theoretical results to 
  thermometry with an NV center interacting with its environment via magnetic-field fluctuations,  the symbols
  \cmark{} and
  \xmark{} indicate whether HS can be achieved in this specific model.}
\label{fig:table}
\end{figure}


We now 
include also spontaneous decay, 
while keeping the same coupling operators $S_\alpha$
which implies the additional condition $\forall_{\alpha\in\{x,y,z\}}\, \braket{\psi_0|S_\alpha|\psi_1}=0$
for a decoherence-free subspace. 
All conditions together then ensure 
that all spin operators act trivially on the subspace $\mathcal C=\t{span}\{\ket{\psi_0},\ket{\psi_1}\}$, i.e., $\forall_{\alpha\in\{x,y,z\}}\,\Pi_\mathcal C S_\alpha \Pi_\mathcal C\propto \Pi_\mathcal C$. 
This cannot be satisfied by any pair of vectors confined to the spin-1 three-dimensional space (see End Matter~\ref{app:example} for a proof). 
However, \autoref{thm:1} guarantees the existence of such a subspace in the Hilbert space extended by an auxiliary system. 
The subspace $\vspan\{\ket{0}\ket{\downarrow},\ket{\psi_-}\ket{\uparrow}\}$,
with orthogonal ancilla states $\ket{\downarrow},\ket{\uparrow}$, satisfies all
conditions and allows for HS.

The seemingly surprising fact that, for low-temperature noise, dressed states achieve HS where HNLS applied to the bare system forbids it has a simple origin: the latter was derived under the assumption that QEC operates on timescales longer than the noise-averaging time, and therefore addresses errors that have already occurred. 
The dressing acts earlier: the noise coupling terms average to zero over the bath correlation time, and no error forms.
This makes the condition for HS less restrictive---aptly captured by the Latin maxim \textit{morbum evitare quam curare facilius est} (``it is easier to prevent than to cure'').

Finally, for the most general bath spectral function, including thermal excitation, the generator $G$ is proportional to the square of one of the coupling operators, $G\propto S_z^2$. As a consequence, \autoref{thm:2} implies that neither decoherence-free subspaces nor quantum error-correction protocols can recover Heisenberg scaling. Determining whether more advanced noise-mitigation techniques could surpass these limits would necessitate a more refined environmental model and a description extending beyond the master equation approach.
Both theory and the discussed example are summarized in \autoref{fig:table}.

{\it Conclusions---}We have determined 
the conditions under which Heisenberg scaling can be achieved in noisy quantum metrology using dressed states generated by a constant control Hamiltonian. 
For dephasing and relaxation in the absence of thermal excitation,  
HS is achievable if and only if the signal generator lies outside the linear span of the system--environment coupling operators---in contrast to the conventional HNLS condition, involving the Lindblad operators together with their pairwise products. A carefully engineered dressed basis can therefore recover HS even where HNLS, evaluated for the undressed system, would rule it out. Our scheme exploits the finite bath correlation time, in line with the general
principle that correlated noise is easier to protect
against~\cite{altherr2021quantum,mann2025quantum,kurdzialek2025universal}. We have illustrated our results using an example inspired by
thermometry with an NV-center subject to magnetic field fluctuations, but our findings are universal, since they are based on algebraic relations between operators.
They can also be applied to other experimental platforms, including trapped ions, Rydberg atoms, and cold atomic gases---not necessarily for quantitative predictions, 
but to identify settings for which an advantage exists.
More, generally, our
results offer a recipe for tailoring control Hamiltonians to achieve desired metrological properties, while also serving as a foundation for future time-dependent sensing protocols. These could include adiabatic approaches for slow control~\cite{albash2012quantum}, non-adiabatic master equations for fast driving~\cite{dann2018time}, as well as periodic driving (Floquet) regimes where the drive frequency explicitly shifts the sampled bath spectrum~\cite[Sec. 8.4]{breuer2002theory}---all of which constitute  natural directions for future research.

{\it Acknowledgments---}We thank Francesco Albarelli for pointing us to relevant literature and Lorenza Viola for insightful discussions on the physical timescales behind the master-equation. Financial support from the Deutsche Forschungsgemeinschaft (DFG) – Project No. 277101999, Collaborative Research Centre (CRC) 183 (project B02) and the Alexander von Humboldt foundation are gratefully acknowledged.

\textit{Note added}---While our manuscript was under review, Ref.~\cite{zeng2026achieving} appeared, 
which discusses a iff conditions for HS for quantum Zeno dynamics (QZD) and dynamical decoupling (DD), obtaining
the same condition as in Lemma~\ref{lem:core}. This is not surprising, since Zeno dynamics may be realized
in the limit of an infinitely strong constant Hamiltonian~\cite{facchi2002quantum,Facchi_2008}.

\bibliography{biblio}

@article{groszkowski2023simple,
  title={Simple master equations for describing driven systems subject to classical non-Markovian noise},
  author={Groszkowski, Peter and Seif, Alireza and Koch, Jens and Clerk, Aashish A},
  journal={Quantum},
  volume={7},
  pages={972},
  year={2023},
  publisher={Verein zur F{\"o}rderung des Open Access Publizierens in den Quantenwissenschaften},
  doi={10.22331/q-2023-04-06-972}
}

@article{zeng2026achieving,
  title={Achieving Heisenberg limit under noisy conditions with quantum Zeno dynamics and dynamical decoupling},
  author={Zeng, Ke and Lahcen, Bakmou and Jiang, Yu and Tan, Kok Chuan},
  journal={arXiv preprint arXiv:2606.13205},
  year={2026},
  doi={10.48550/arXiv.2606.13205}
}

@article{Facchi_2008,
doi = {10.1088/1751-8113/41/49/493001},
url = {https://doi.org/10.1088/1751-8113/41/49/493001},
year = {2008},
month = {oct},
publisher = {},
volume = {41},
number = {49},
pages = {493001},
author = {Facchi, P and Pascazio, S},
title = {Quantum Zeno dynamics: mathematical and physical aspects},
journal = {Journal of Physics A: Mathematical and Theoretical}
}

@article{facchi2002quantum,
  title={Quantum zeno subspaces},
  author={Facchi, Paolo and Pascazio, Saverio},
  journal={Physical review letters},
  volume={89},
  number={8},
  pages={080401},
  year={2002},
  publisher={APS},
  doi={10.1103/PhysRevLett.89.080401}
}

@article{Burgarth2022oneboundtorulethem,
  doi = {10.22331/q-2022-06-14-737},
  url = {https://doi.org/10.22331/q-2022-06-14-737},
  title = {One bound to rule them all: from {A}diabatic to {Z}eno},
  author = {Burgarth, Daniel and Facchi, Paolo and Gramegna, Giovanni and Yuasa, Kazuya},
  journal = {{Quantum}},
  issn = {2521-327X},
  publisher = {{Verein zur F{\"{o}}rderung des Open Access Publizierens in den Quantenwissenschaften}},
  volume = {6},
  pages = {737},
  month = jun,
  year = {2022}
}

@article{gaebel2006room,
  title={Room-temperature coherent coupling of single spins in diamond},
  author={Gaebel, Torsten and Domhan, Michael and Popa, Iulian and Wittmann, Christoffer and Neumann, Philipp and Jelezko, Fedor and Rabeau, James R and Stavrias, Nikolas and Greentree, Andrew D and Prawer, Steven and others},
  journal={Nature Physics},
  volume={2},
  number={6},
  pages={408--413},
  year={2006},
  publisher={Nature Publishing Group UK London},
  doi={10.1038/nphys318}
}

@article{brion2007adiabatic,
  title={Adiabatic elimination in a lambda system},
  author={Brion, Etienne and Pedersen, Line Hjortsh{\o}j and M{\o}lmer, Klaus},
  journal={Journal of Physics A: Mathematical and Theoretical},
  volume={40},
  number={5},
  pages={1033--1043},
  year={2007},
  doi={10.1088/1751-8113/40/5/011}
}

@article{riberi2025optimal,
  title={Optimal asymptotic precision bounds for nonlinear quantum metrology under collective dephasing},
  author={Riberi, Francisco and Viola, Lorenza},
  journal={APL Quantum},
  volume={2},
  number={2},
  year={2025},
  publisher={AIP Publishing},
  doi={10.1063/5.0255629}
}

@article{riberi2026precision,
  title={Precision bounds for frequency estimation under collective dephasing and open-loop control},
  author={Riberi, Francisco and Paz-Silva, Gerardo and Viola, Lorenza},
  journal={arXiv preprint arXiv:2603.23804},
  year={2026},
  doi={10.48550/arXiv.2603.23804}
}

@article{riberi2022frequency,
  title={Frequency estimation under non-Markovian spatially correlated quantum noise},
  author={Riberi, Francisco and Norris, Leigh M and Beaudoin, F{\'e}lix and Viola, Lorenza},
  journal={New Journal of Physics},
  volume={24},
  number={10},
  pages={103011},
  year={2022},
  publisher={IOP Publishing},
  doi={10.1088/1367-2630/ac92a2}
}

@article{zhou2020quantum,
  title={Quantum metrology with strongly interacting spin systems},
  author={Zhou, Hengyun and Choi, Joonhee and Choi, Soonwon and Landig, Renate and Douglas, Alexander M and Isoya, Junichi and Jelezko, Fedor and Onoda, Shinobu and Sumiya, Hitoshi and Cappellaro, Paola and others},
  journal={Physical review X},
  volume={10},
  number={3},
  pages={031003},
  year={2020},
  publisher={APS},
  doi={10.1103/PhysRevX.10.031003}
}

@article{hernandez2024optimal,
  title={Optimal control of a quantum sensor: A fast algorithm based on an analytic solution},
  author={Hern{\'a}ndez-G{\'o}mez, Santiago and Balducci, Federico and Fasiolo, Giovanni and Cappellaro, Paola and Fabbri, Nicole and Scardicchio, Antonello},
  journal={SciPost Physics},
  volume={17},
  number={1},
  pages={004},
  year={2024},
  doi={10.21468/SciPostPhys.17.1.004}
}

@article{puig2025dynamical,
  title={From dynamical to steady-state many-body metrology: Precision limits and their attainability with two-body interactions},
  author={Puig, Ricard and Sekatski, Pavel and Erdman, Paolo Andrea and Abiuso, Paolo and Calsamiglia, John and Perarnau-Llobet, Mart{\'\i}},
  journal={PRX Quantum},
  volume={6},
  number={3},
  pages={030309},
  year={2025},
  publisher={APS},
  doi={10.1103/PRXQuantum.6.030309}
}

@article{beaudoin2011dissipation,
  title={Dissipation and ultrastrong coupling in circuit QED},
  author={Beaudoin, F{\'e}lix and Gambetta, Jay M and Blais, A},
  journal={Physical Review A—Atomic, Molecular, and Optical Physics},
  volume={84},
  number={4},
  pages={043832},
  year={2011},
  publisher={APS},
  doi={10.1103/PhysRevA.84.043832}
}

@article{wang2019heisenberg,
  title={Heisenberg-limited single-mode quantum metrology in a superconducting circuit},
  author={Wang, Weiting and Wu, Yukai and Ma, Yuwei and Cai, Weizhou and Hu, Ling and Mu, Xianghao and Xu, Yuan and Chen, Zi-Jie and Wang, Haiyan and Song, YP and others},
  journal={Nature communications},
  volume={10},
  number={1},
  pages={4382},
  year={2019},
  publisher={Nature Publishing Group UK London},
  doi={10.1038/s41467-019-12290-7}
}

@article{liu2017quantum,
  title={Quantum parameter estimation with optimal control},
  author={Liu, Jing and Yuan, Haidong},
  journal={Physical Review A},
  volume={96},
  number={1},
  pages={012117},
  year={2017},
  publisher={APS},
  doi={10.1103/PhysRevA.96.012117}
}

@article{kurdzialek2025quantum,
  title={Quantum metrology using quantum combs and tensor network formalism},
  author={Kurdzia{\l}ek, Stanis{\l}aw and Dulian, Piotr and Majsak, Joanna and Chakraborty, Sagnik and Demkowicz-Dobrza{\'n}ski, Rafa{\l}},
  journal={New Journal of Physics},
  volume={27},
  number={1},
  pages={013019},
  year={2025},
  publisher={IOP Publishing},
  doi={10.1088/1367-2630/ada8d1}
}

@inproceedings{ma2021adaptive,
  title={Adaptive circuit learning for quantum metrology},
  author={Ma, Ziqi and Gokhale, Pranav and Zheng, Tian-Xing and Zhou, Sisi and Yu, Xiaofei and Jiang, Liang and Maurer, Peter and Chong, Frederic T},
  booktitle={2021 IEEE International Conference on Quantum Computing and Engineering (QCE)},
  pages={419--430},
  year={2021},
  organization={IEEE},
  doi={10.1109/QCE52317.2021.00063}
}

@article{toth2014quantum,
  title={Quantum metrology from a quantum information science perspective},
  author={T{\'o}th, G{\'e}za and Apellaniz, Iagoba},
  journal={Journal of Physics A: Mathematical and Theoretical},
  volume={47},
  number={42},
  pages={424006},
  year={2014},
  publisher={IOP Publishing},
  doi={10.1088/1751-8113/47/42/424006}
}

@article{marciniak2022optimal,
  title={Optimal metrology with programmable quantum sensors},
  author={Marciniak, Christian D and Feldker, Thomas and Pogorelov, Ivan and Kaubruegger, Raphael and Vasilyev, Denis V and van Bijnen, Rick and Schindler, Philipp and Zoller, Peter and Blatt, Rainer and Monz, Thomas},
  journal={Nature},
  volume={603},
  number={7902},
  pages={604--609},
  year={2022},
  publisher={Nature Publishing Group UK London},
  doi={10.1038/s41586-022-04435-4}
}

@article{huang2024entanglement,
  title={Entanglement-enhanced quantum metrology: From standard quantum limit to Heisenberg limit},
  author={Huang, Jiahao and Zhuang, Min and Lee, Chaohong},
  journal={Applied Physics Reviews},
  volume={11},
  number={3},
  year={2024},
  publisher={AIP Publishing},
  doi={10.1063/5.0204102}
}

@article{liu2024fully,
  title={Fully-optimized quantum metrology: framework, tools, and applications},
  author={Liu, Qiushi and Hu, Zihao and Yuan, Haidong and Yang, Yuxiang},
  journal={Advanced Quantum Technologies},
  volume={7},
  number={12},
  pages={2400094},
  year={2024},
  publisher={Wiley Online Library},
  doi={10.1002/qute.202400094}
}

@article{ticozzi2014quantum,
  title={Quantum resources for purification and cooling: fundamental limits and opportunities},
  author={Ticozzi, Francesco and Viola, Lorenza},
  journal={Scientific reports},
  volume={4},
  number={1},
  pages={5192},
  year={2014},
  publisher={Nature Publishing Group UK London},
  doi={10.1038/srep05192}
}

@article{ticozzi2017quantum,
  title={Quantum and classical resources for unitary design of open-system evolutions},
  author={Ticozzi, Francesco and Viola, Lorenza},
  journal={Quantum Science and Technology},
  volume={2},
  number={3},
  pages={034001},
  year={2017},
  publisher={IOP Publishing},
  doi={10.1088/2058-9565/aa722a}
}

@article{basilewitsch2021fundamental,
  title={Fundamental bounds on qubit reset},
  author={Basilewitsch, Daniel and Fischer, Jonas and Reich, Daniel M and Sugny, Dominique and Koch, Christiane P},
  journal={Physical Review Research},
  volume={3},
  number={1},
  pages={013110},
  year={2021},
  publisher={APS},
  doi={10.1103/PhysRevResearch.3.013110}
}

@article{altherr2021quantum,
  title={Quantum metrology for non-Markovian processes},
  author={Altherr, Anian and Yang, Yuxiang},
  journal={Physical Review Letters},
  volume={127},
  number={6},
  pages={060501},
  year={2021},
  publisher={APS},
  doi={10.1103/PhysRevLett.127.060501}
}

@article{casacio2021quantum,
  title={Quantum-enhanced nonlinear microscopy},
  author={Casacio, Catxere A and Madsen, Lars S and Terrasson, Alex and Waleed, Muhammad and Barnscheidt, Kai and Hage, Boris and Taylor, Michael A and Bowen, Warwick P},
  journal={Nature},
  volume={594},
  number={7862},
  pages={201--206},
  year={2021},
  publisher={Nature Publishing Group UK London},
  doi={10.1038/s41586-021-03528-w}
}

@article{bothwell2022resolving,
  title={Resolving the gravitational redshift across a millimetre-scale atomic sample},
  author={Bothwell, Tobias and Kennedy, Colin J and Aeppli, Alexander and Kedar, Dhruv and Robinson, John M and Oelker, Eric and Staron, Alexander and Ye, Jun},
  journal={Nature},
  volume={602},
  number={7897},
  pages={420--424},
  year={2022},
  publisher={Nature Publishing Group UK London},
  doi={10.1038/s41586-021-04349-7}
}

@article{zhang2024frequency,
  title={Frequency ratio of the 229mTh nuclear isomeric transition and the 87Sr atomic clock},
  author={Zhang, Chuankun and Ooi, Tian and Higgins, Jacob S and Doyle, Jack F and von der Wense, Lars and Beeks, Kjeld and Leitner, Adrian and Kazakov, Georgy A and Li, Peng and Thirolf, Peter G and others},
  journal={Nature},
  volume={633},
  number={8028},
  pages={63--70},
  year={2024},
  publisher={Nature Publishing Group UK London},
  doi={10.1038/s41586-024-07839-6}
}

@article{campbell2017fermi,
  title={A Fermi-degenerate three-dimensional optical lattice clock},
  author={Campbell, Sara L and Hutson, RB and Marti, GE and Goban, Akihisa and Darkwah Oppong, Nelson and McNally, RL and Sonderhouse, Lindsay and Robinson, JM and Zhang, W and Bloom, BJ and others},
  journal={Science},
  volume={358},
  number={6359},
  pages={90--94},
  year={2017},
  publisher={American Association for the Advancement of Science},
  doi={10.1126/science.aam5538}
}

@article{dann2018time,
  title={Time-dependent Markovian quantum master equation},
  author={Dann, Roie and Levy, Amikam and Kosloff, Ronnie},
  journal={Physical Review A},
  volume={98},
  number={5},
  pages={052129},
  year={2018},
  publisher={APS},
  doi={10.1103/PhysRevA.98.052129}
}

@article{de2013quantum,
  title={Quantum parameter estimation affected by unitary disturbance},
  author={De Pasquale, Antonella and Rossini, Davide and Facchi, Paolo and Giovannetti, Vittorio},
  journal={Phys. Rev. A},
  volume={88},
  number={5},
  pages={052117},
  year={2013},
  publisher={APS},
doi={10.1103/PhysRevA.88.052117}
}

@article{ludlow2015optical,
  title={Optical atomic clocks},
  author={Ludlow, Andrew D and Boyd, Martin M and Ye, Jun and Peik, Ekkehard and Schmidt, Piet O},
  journal={Reviews of Modern Physics},
  volume={87},
  number={2},
  pages={637--701},
  year={2015},
  publisher={APS},
  doi={10.1103/RevModPhys.87.637}
}

@article{nicholson2015systematic,
  title={Systematic evaluation of an atomic clock at 2$\times$ 10- 18 total uncertainty},
  author={Nicholson, Travis L and Campbell, SL and Hutson, RB and Marti, G Edward and Bloom, BJ and McNally, Rees L and Zhang, Wei and Barrett, MD and Safronova, Marianna S and Strouse, GF and others},
  journal={Nature communications},
  volume={6},
  number={1},
  pages={6896},
  year={2015},
  publisher={Nature Publishing Group UK London},
  doi={10.1038/ncomms7896}
}

@article{bloom2014optical,
  title={An optical lattice clock with accuracy and stability at the 10- 18 level},
  author={Bloom, BJ and Nicholson, TL and Williams, JR and Campbell, SL and Bishof, M and Zhang, X and Zhang, W and Bromley, SL and Ye, J},
  journal={Nature},
  volume={506},
  number={7486},
  pages={71--75},
  year={2014},
  publisher={Nature Publishing Group UK London},
  doi={10.1038/nature12941}
}

@article{bongs2019taking,
  title={Taking atom interferometric quantum sensors from the laboratory to real-world applications},
  author={Bongs, Kai and Holynski, Michael and Vovrosh, Jamie and Bouyer, Philippe and Condon, Gabriel and Rasel, Ernst and Schubert, Christian and Schleich, Wolfgang P and Roura, Albert},
  journal={Nature Reviews Physics},
  volume={1},
  number={12},
  pages={731--739},
  year={2019},
  publisher={Nature Publishing Group UK London},
  doi={10.1038/s42254-019-0117-4}
}

@article{rosi2014precision,
  title={Precision measurement of the Newtonian gravitational constant using cold atoms},
  author={Rosi, G and Sorrentino, F and Cacciapuoti, L and Prevedelli, Marco and Tino, GM},
  journal={Nature},
  volume={510},
  number={7506},
  pages={518--521},
  year={2014},
  publisher={Nature Publishing Group UK London},
  doi={10.1038/nature13433}
}

@article{kasevich1991atomic,
  title={Atomic interferometry using stimulated Raman transitions},
  author={Kasevich, Mark and Chu, Steven},
  journal={Physical review letters},
  volume={67},
  number={2},
  pages={181},
  year={1991},
  publisher={APS},
  doi={10.1103/PhysRevLett.67.181}
}

@article{tse2019quantum,
  title={Quantum-enhanced advanced LIGO detectors in the era of gravitational-wave astronomy},
  author={Tse, Maggie and Yu, Haocun and Kijbunchoo, Nutsinee and Fernandez-Galiana, A and Dupej, P and Barsotti, L and Blair, CD and Brown, DD and Dwyer, SE ea and Effler, A and others},
  journal={Physical Review Letters},
  volume={123},
  number={23},
  pages={231107},
  year={2019},
  publisher={APS},
  doi={10.1103/PhysRevLett.123.231107}
}

@article{acernese2019increasing,
  title={Increasing the astrophysical reach of the advanced virgo detector via the application of squeezed vacuum states of light},
  author={Acernese, Fausto and Agathos, M and Aiello, L and Allocca, A and Amato, A and Ansoldi, S and Antier, S and Ar{\`e}ne, M and Arnaud, N and Ascenzi, S and others},
  journal={Physical review letters},
  volume={123},
  number={23},
  pages={231108},
  year={2019},
  publisher={APS},
  doi={10.1103/PhysRevLett.123.231108}
}

@article{kucsko2013nanometre,
  title={Nanometre-scale thermometry in a living cell},
  author={Kucsko, Georg and Maurer, Peter C and Yao, Norman Ying and Kubo, MICHAEL and Noh, Hyun Jong and Lo, Po Kam and Park, Hongkun and Lukin, Mikhail D},
  journal={Nature},
  volume={500},
  number={7460},
  pages={54--58},
  year={2013},
  publisher={Nature Publishing Group UK London},
  doi={10.1038/nature12373}
}

@article{barry2020sensitivity,
  title={Sensitivity optimization for NV-diamond magnetometry},
  author={Barry, John F and Schloss, Jennifer M and Bauch, Erik and Turner, Matthew J and Hart, Connor A and Pham, Linh M and Walsworth, Ronald L},
  journal={Reviews of Modern Physics},
  volume={92},
  number={1},
  pages={015004},
  year={2020},
  publisher={APS},
  doi={10.1103/RevModPhys.92.015004}
}

@article{brida2010experimental,
  title={Experimental realization of sub-shot-noise quantum imaging},
  author={Brida, Giorgio and Genovese, Marco and Ruo Berchera, Ivano},
  journal={Nature Photonics},
  volume={4},
  number={4},
  pages={227--230},
  year={2010},
  publisher={Nature Publishing Group UK London},
  doi={10.1038/nphoton.2010.29}
}

@article{cheiney2018navigation,
  title={Navigation-compatible hybrid quantum accelerometer using a Kalman filter},
  author={Cheiney, Pierrick and Fouch{\'e}, Lauriane and Templier, Simon and Napolitano, Fabien and Battelier, Baptiste and Bouyer, Philippe and Barrett, Brynle},
  journal={Physical Review Applied},
  volume={10},
  number={3},
  pages={034030},
  year={2018},
  publisher={APS},
  doi={10.1103/PhysRevApplied.10.034030}
}

@article{kurdzialek2025universal,
  title={Universal bounds for quantum metrology in the presence of correlated noise},
  author={Kurdzia{\l}ek, Stanis{\l}aw and Albarelli, Francesco and Demkowicz-Dobrza{\'n}ski, Rafa{\l}},
  journal={Physical Review Letters},
  volume={135},
  number={13},
  pages={130801},
  year={2025},
  publisher={APS},
doi={10.1103/jy3v-wkcb}
}

@article{arrad2014increasing,
  title={Increasing sensing resolution with error correction},
  author={Arrad, Gilad and Vinkler, Yuval and Aharonov, Dorit and Retzker, Alex},
  journal={Physical review letters},
  volume={112},
  number={15},
  pages={150801},
  year={2014},
  publisher={APS},
doi={10.1103/PhysRevLett.112.150801}
}

@article{haase2018fundamental,
  title={Fundamental limits to frequency estimation: a comprehensive microscopic perspective},
  author={Haase, Jan F and Smirne, Andrea and Ko{\l}ody{\'n}ski, Jan and Demkowicz-Dobrza{\'n}ski, Rafa{\l} and Huelga, Susana F},
  journal={New Journal of Physics},
  volume={20},
  number={5},
  pages={053009},
  year={2018},
  publisher={IOP Publishing},
doi={10.1088/1367-2630/aab67f}
}

@article{mann2025quantum,
  title={Quantum error-corrected non-Markovian metrology},
  author={Mann, Zachary and Cao, Ningping and Laflamme, Raymond and Zhou, Sisi},
  journal={PRX Quantum},
  volume={6},
  number={3},
  pages={030321},
  year={2025},
  publisher={APS},
doi={10.1103/wfyl-wtz3}
}

@article{tabuchi2023temperature,
  title={Temperature sensing with RF-dressed states of nitrogen-vacancy centers in diamond},
  author={Tabuchi, Hibiki and Matsuzaki, Yuichiro and Furuya, Noboru and Nakano, Yuta and Watanabe, Hideyuki and Tokuda, Norio and Mizuochi, Norikazu and Ishi-Hayase, Junko},
  journal={Journal of Applied Physics},
  volume={133},
  number={2},
  pages = {024401},
  year={2023},
  publisher={AIP Publishing},
doi={10.1063/5.0129706}
}

@article{yun2021temperature,
  title={Temperature Selective Thermometry with Sub-Microsecond Time Resolution Using Dressed-Spin States in Diamond},
  author={Yun, Jiwon and Kim, Kiho and Park, Sungjoon and Kim, Dohun},
  journal={Advanced Quantum Technologies},
  volume={4},
  number={11},
  pages={2100084},
  year={2021},
  publisher={Wiley Online Library},
doi={10.1002/qute.202100084}
}

@article{hirose2012continuous,
  title={Continuous dynamical decoupling magnetometry},
  author={Hirose, Masashi and Aiello, Clarice D and Cappellaro, Paola},
  journal={Physical Review A—Atomic, Molecular, and Optical Physics},
  volume={86},
  number={6},
  pages={062320},
  year={2012},
  publisher={APS},
doi={10.1103/PhysRevA.86.062320}
}

@article{beaver2024selective,
  title={Selective temperature sensing in nanodiamonds using dressed states},
  author={Beaver, Nathaniel M and Stevenson, Paul},
  journal={Advanced Quantum Technologies},
  volume={7},
  number={12},
  pages={2400271},
  year={2024},
  publisher={Wiley Online Library},
doi={10.1002/qute.202400271}
}

@article{xu2012coherence,
  title={Coherence-protected quantum gate by continuous dynamical decoupling in diamond},
  author={Xu, Xiangkun and Wang, Zixiang and Duan, Changkui and Huang, Pu and Wang, Pengfei and Wang, Ya and Xu, Nanyang and Kong, Xi and Shi, Fazhan and Rong, Xing and others},
  journal={Physical review letters},
  volume={109},
  number={7},
  pages={070502},
  year={2012},
  publisher={APS},
doi={10.1103/PhysRevLett.109.070502}
}

@article{golter2014protecting,
  title={Protecting a solid-state spin from decoherence using dressed spin states},
  author={Golter, D Andrew and Baldwin, Thomas K and Wang, Hailin},
  journal={Physical review letters},
  volume={113},
  number={23},
  pages={237601},
  year={2014},
  publisher={APS},
doi={10.1103/PhysRevLett.113.237601}
}

@article{timoney2011quantum,
  title={Quantum gates and memory using microwave-dressed states},
  author={Timoney, N and Baumgart, I and Johanning, M and Var{\'o}n, AF and Plenio, Martin B and Retzker, A and Wunderlich, Ch},
  journal={Nature},
  volume={476},
  number={7359},
  pages={185--188},
  year={2011},
  publisher={Nature Publishing Group UK London},
doi={10.1038/nature10319}
}

@article{paz2014general,
  title={General transfer-function approach to noise filtering in open-loop quantum control},
  author={Paz-Silva, Gerardo A and Viola, Lorenza},
  journal={Physical review letters},
  volume={113},
  number={25},
  pages={250501},
  year={2014},
  publisher={APS},
doi={PhysRevLett.113.250501}
}

@article{soare2014experimental,
  title={Experimental noise filtering by quantum control},
  author={Soare, A and Ball, H and Hayes, D and Sastrawan, J and Jarratt, MC and McLoughlin, JJ and Zhen, X and Green, TJ and Biercuk, MJ},
  journal={Nature Physics},
  volume={10},
  number={11},
  pages={825--829},
  year={2014},
  publisher={Nature Publishing Group UK London},
doi={10.1038/nphys3115}
}

@article{green2013arbitrary,
  title={Arbitrary quantum control of qubits in the presence of universal noise},
  author={Green, Todd J and Sastrawan, Jarrah and Uys, Hermann and Biercuk, Michael J},
  journal={New Journal of Physics},
  volume={15},
  number={9},
  pages={095004},
  year={2013},
  publisher={IOP Publishing},
doi={10.1088/1367-2630/15/9/095004}
}

@article{chen2006geometric,
  title={Geometric continuous dynamical decoupling with bounded controls},
  author={Chen, Pochung},
  journal={Physical Review A—Atomic, Molecular, and Optical Physics},
  volume={73},
  number={2},
  pages={022343},
  year={2006},
  publisher={APS},
doi={10.1103/PhysRevA.73.022343}
}

@article{hanson2008coherent,
  title={Coherent dynamics of a single spin interacting with an adjustable spin bath},
  author={Hanson, R and Dobrovitski, VV and Feiguin, AE and Gywat, O and Awschalom, DD},
  journal={Science},
  volume={320},
  number={5874},
  pages={352--355},
  year={2008},
  publisher={American Association for the Advancement of Science},
doi={10.1126/science.1155400}
}

@article{rondin2014magnetometry,
  title={Magnetometry with nitrogen-vacancy defects in diamond},
  author={Rondin, Lo{\"\i}c and Tetienne, Jean-Philippe and Hingant, Thomas and Roch, Jean-Fran{\c{c}}ois and Maletinsky, Patrick and Jacques, Vincent},
  journal={Reports on progress in physics},
  volume={77},
  number={5},
  pages={056503},
  year={2014},
  publisher={IOP Publishing},
doi={10.1088/0034-4885/77/5/056503}
}

@article{dolde2011electric,
  title={Electric-field sensing using single diamond spins},
  author={Dolde, Florian and Fedder, Helmut and Doherty, Marcus W and N{\"o}bauer, Tobias and Rempp, Florian and Balasubramanian, Gopalakrishnan and Wolf, Thomas and Reinhard, Friedemann and Hollenberg, Lloyd CL and Jelezko, Fedor and others},
  journal={Nature Physics},
  volume={7},
  number={6},
  pages={459--463},
  year={2011},
  publisher={Nature Publishing Group UK London},
doi={10.1038/NPHYS1969}
}

@article{doherty2013nitrogen,
  title={The nitrogen-vacancy colour centre in diamond},
  author={Doherty, Marcus W and Manson, Neil B and Delaney, Paul and Jelezko, Fedor and Wrachtrup, J{\"o}rg and Hollenberg, Lloyd CL},
  journal={Physics Reports},
  volume={528},
  number={1},
  pages={1--45},
  year={2013},
  publisher={Elsevier},
doi={10.1016/j.physrep.2013.02.001}
}

@article{albash2012quantum,
  title={Quantum adiabatic Markovian master equations},
  author={Albash, Tameem and Boixo, Sergio and Lidar, Daniel A and Zanardi, Paolo},
  journal={New Journal of Physics},
  volume={14},
  number={12},
  pages={123016},
  year={2012},
  publisher={IOP Publishing},
doi={10.1088/1367-2630/14/12/123016}
}

@article{chen2024quantum,
  title={Quantum metrology enhanced by leveraging informative noise with error correction},
  author={Chen, Hongzhen and Chen, Yu and Liu, Jing and Miao, Zibo and Yuan, Haidong},
  journal={Physical Review Letters},
  volume={133},
  number={19},
  pages={190801},
  year={2024},
  publisher={APS},
doi={10.1103/PhysRevLett.133.190801}
}

@article{sekatski2016dynamical,
  title={Dynamical decoupling leads to improved scaling in noisy quantum metrology},
  author={Sekatski, Pavel and Skotiniotis, Michalis and D{\"u}r, Wolfgang},
  journal={New Journal of Physics},
  volume={18},
  number={7},
  pages={073034},
  year={2016},
  publisher={IOP Publishing},
doi={10.1088/1367-2630/18/7/073034}
}

@article{macieszczak2015zeno,
  title={Zeno limit in frequency estimation with non-Markovian environments},
  author={Macieszczak, Katarzyna},
  journal={Physical Review A},
  volume={92},
  number={1},
  pages={010102},
  year={2015},
  publisher={APS},
doi={https://doi.org/10.1103/PhysRevA.92.010102}
}

@article{das2025universal,
  title={Universal time scalings of sensitivity in Markovian quantum metrology},
  author={Das, Arpan and G{\'o}recki, Wojciech and Demkowicz-Dobrza{\'n}ski, Rafa{\l}},
  journal={Physical Review A},
  volume={111},
  number={2},
  pages={L020403},
  year={2025},
  publisher={APS},
doi={10.1103/PhysRevA.111.L020403}
}

@article{wan2022bounds,
  title = {Bounds on adaptive quantum metrology under Markovian noise},
  author = {Wan, Kianna and Lasenby, Robert},
  journal = {Phys. Rev. Res.},
  volume = {4},
  issue = {3},
  pages = {033092},
  numpages = {18},
  year = {2022},
  month = {Aug},
  publisher = {American Physical Society},
  doi = {10.1103/PhysRevResearch.4.033092},
  url = {https://link.aps.org/doi/10.1103/PhysRevResearch.4.033092}
}

@article{kurdzialek2022using,
  title = {Using Adaptiveness and Causal Superpositions Against Noise in Quantum Metrology},
  author = {Kurdzia\l{}ek, Stanis\l{}aw and G\'orecki, Wojciech and Albarelli, Francesco and Demkowicz-Dobrza\ifmmode \acute{n}\else \'{n}\fi{}ski, Rafa\l{}},
  journal = {Phys. Rev. Lett.},
  volume = {131},
  issue = {9},
  pages = {090801},
  numpages = {7},
  year = {2023},
  month = {Aug},
  publisher = {American Physical Society},
  doi = {10.1103/PhysRevLett.131.090801},
  url = {https://link.aps.org/doi/10.1103/PhysRevLett.131.090801}
}

@article{zhou2021asymptotic,
  title = {Asymptotic Theory of Quantum Channel Estimation},
  author = {Zhou, Sisi and Jiang, Liang},
  journal = {PRX Quantum},
  volume = {2},
  issue = {1},
  pages = {010343},
  numpages = {25},
  year = {2021},
  month = {Mar},
  publisher = {American Physical Society},
  doi = {10.1103/PRXQuantum.2.010343},
  url = {https://link.aps.org/doi/10.1103/PRXQuantum.2.010343}
}

@incollection{demkowicz2015optical,
title = {Chapter Four - Quantum Limits in Optical Interferometry},
editor = {E. Wolf},
series = {Progress in Optics},
publisher = {Elsevier},
volume = {60},
pages = {345-435},
year = {2015},
issn = {0079-6638},
doi = {https://doi.org/10.1016/bs.po.2015.02.003},
url = {https://www.sciencedirect.com/science/article/pii/S0079663815000049},
author = {Rafal Demkowicz-Dobrzański and Marcin Jarzyna and Jan Kołodyński}
}

@Book{breuer2002theory,
  Title                    = {The theory of open quantum systems},
  Author                   = {Breuer, Heinz-Peter and Petruccione, Francesco and others},
  Publisher                = {Oxford University Press on Demand},
  Year                     = {2002}
}

@Article{caves1981quantum,
  Title                    = {Quantum-mechanical noise in an interferometer},
  Author                   = {Caves, Carlton M},
  Journal                  = {Phys. Rev. D},
  Year                     = {1981},
  Number                   = {8},
  Pages                    = {1693},
  Volume                   = {23},
  Publisher                = {APS},
  url = {https://journals.aps.org/prd/pdf/10.1103/PhysRevD.23.1693}
}

@Article{chin2012quantum,
  Title                    = {Quantum metrology in non-Markovian environments},
  Author                   = {Chin, Alex W and Huelga, Susana F and Plenio, Martin B},
  Journal                  = {Phys. Rev. Lett.},
  Year                     = {2012},
  Number                   = {23},
  Pages                    = {233601},
  Volume                   = {109},
  Publisher                = {APS},
  url = {https://journals.aps.org/prl/pdf/10.1103/PhysRevLett.109.233601}
}

@Article{degen2017quantum,
  Title                    = {Quantum sensing},
  Author                   = {Degen, Christian L and Reinhard, F and Cappellaro, P},
  Journal                  = {Rev. Mod. Phys.},
  Year                     = {2017},
  Number                   = {3},
  Pages                    = {035002},
  Volume                   = {89},
  Publisher                = {APS},
  url = {https://journals.aps.org/rmp/pdf/10.1103/RevModPhys.89.035002}
}

@Article{demkowicz2012elusive,
  Title                    = {The elusive Heisenberg limit in quantum-enhanced metrology},
  Author                   = {Demkowicz-Dobrza{\'n}ski, Rafa{\l} and Ko{\l}ody{\'n}ski, Jan and Gu{\c{t}}{\u{a}}, M{\u{a}}d{\u{a}}lin},
  Journal                  = {Nat. Commun.},
  Year                     = {2012},
  Pages                    = {1063},
  Volume                   = {3},
  Publisher                = {Nature Publishing Group},
  url = {https://www.nature.com/articles/ncomms2067}
}

@Article{demkowicz2014using,
  Title                    = {Using entanglement against noise in quantum metrology},
  Author                   = {Demkowicz-Dobrza{\'n}ski, Rafal and Maccone, Lorenzo},
  Journal                  = {Phys. Rev. Lett.},
  Year                     = {2014},
  Number                   = {25},
  Pages                    = {250801},
  Volume                   = {113},
  Publisher                = {APS},
  url = {https://journals.aps.org/prl/pdf/10.1103/PhysRevLett.113.250801}
}

@Article{demkowicz2017adaptive,
  Title                    = {Adaptive Quantum Metrology under General Markovian Noise},
  Author                   = {Demkowicz-Dobrza\ifmmode \acute{n}\else \'{n}\fi{}ski, Rafa\l{} and Czajkowski, Jan and Sekatski, Pavel},
  Journal                  = {Phys. Rev. X},
  Year                     = {2017},
  Pages                    = {041009},
  Volume                   = {7},
  DOI                      = {10.1103/PhysRevX.7.041009},
  Issue                    = {4},
  Numpages                 = {15},
  Publisher                = {American Physical Society},
  url = {https://journals.aps.org/prx/pdf/10.1103/PhysRevX.7.041009}
}

@Article{escher2011general,
  Title                    = {General framework for estimating the ultimate precision limit in noisy quantum-enhanced metrology},
  Author                   = {Escher, BM and de Matos Filho, RL and Davidovich, L},
  Journal                  = {Nat. Phys.},
  Year                     = {2011},
  Number                   = {5},
  Pages                    = {406--411},
  Volume                   = {7},
  Publisher                = {Nature Research},
  url = {https://www.nature.com/articles/nphys1958}
}

@Article{fujiwara2008fibre,
  Title                    = {A fibre bundle over manifolds of quantum channels and its application to quantum statistics},
  Author                   = {Fujiwara, Akio and Imai, Hiroshi},
  Journal                  = {J. Phys. A: Math. Theor.},
  Year                     = {2008},
  Number                   = {25},
  Pages                    = {255304},
  Volume                   = {41},
  Publisher                = {IOP Publishing},
  url = {http://iopscience.iop.org/article/10.1088/1751-8113/41/25/255304/pdf}
}

@Article{giovannetti2006quantum,
  Title                    = {Quantum metrology},
  Author                   = {Giovannetti, Vittorio and Lloyd, Seth and Maccone, Lorenzo},
  Journal                  = {Phys. Rev. Lett.},
  Year                     = {2006},
  Number                   = {1},
  Pages                    = {010401},
  Volume                   = {96},
  Publisher                = {APS},
  url = {https://journals.aps.org/prl/pdf/10.1103/PhysRevLett.96.010401}
}

@Article{giovannetti2011advances,
  Title                    = {Advances in quantum metrology},
  Author                   = {Giovannetti, Vittorio and Lloyd, Seth and Maccone, Lorenzo},
  Journal                  = {Nat. Photonics},
  Year                     = {2011},
  Number                   = {4},
  Pages                    = {222--229},
  Volume                   = {5},
  Publisher                = {Nature Research},
  url = {https://www.nature.com/articles/nphoton.2011.35}
}

@Article{knysh2014true,
  Title                    = {True limits to precision via unique quantum probe},
  Author                   = {Knysh, Sergey I and Chen, Edward H and Durkin, Gabriel A},
  Journal                  = {arXiv:1402.0495},
  Year                     = {2014},
  url = {https://arxiv.org/pdf/1402.0495}
}

@Article{kolodynski2013efficient,
  Title                    = {Efficient tools for quantum metrology with uncorrelated noise},
  Author                   = {Ko{\l}ody{\'n}ski, Jan and Demkowicz-Dobrza{\'n}ski, Rafa{\l}},
  Journal                  = {New J. Phys.},
  Year                     = {2013},
  Number                   = {7},
  Pages                    = {073043},
  Volume                   = {15},
  Publisher                = {IOP Publishing},
  url = {http://iopscience.iop.org/article/10.1088/1367-2630/15/7/073043/meta}
}

@article{layden2018ancilla,
  title = {Ancilla-Free Quantum Error Correction Codes for Quantum Metrology},
  author = {Layden, David and Zhou, Sisi and Cappellaro, Paola and Jiang, Liang},
  journal = {Phys. Rev. Lett.},
  volume = {122},
  issue = {4},
  pages = {040502},
  numpages = {6},
  year = {2019},
  month = {Jan},
  publisher = {American Physical Society},
  doi = {10.1103/PhysRevLett.122.040502},
  url = {https://link.aps.org/doi/10.1103/PhysRevLett.122.040502}
}

@Article{Paris2009,
  Title                    = {Quantum Estimation for Quantum Technologies},
  Author                   = {Paris, Mattteo G. A.},
  Journal                  = {Int. J. Quantum Inf.},
  Year                     = {2009},
  Pages                    = {125-137},
  Volume                   = {07},
  DOI                      = {10.1142/S0219749909004839},
  url={https://doi.org/10.1142/S0219749909004839}
}

@Article{Pezze2018,
  Title                    = {Quantum metrology with nonclassical states of atomic ensembles},
  Author                   = {Pezz\`e, Luca and Smerzi, Augusto and Oberthaler, Markus K. and Schmied, Roman and Treutlein, Philipp},
  Journal                  = {Rev. Mod. Phys.},
  Year                     = {2018},
  Month                    = {Sep},
  Pages                    = {035005},
  Volume                   = {90},
  DOI                      = {10.1103/RevModPhys.90.035005},
  Issue                    = {3},
  Numpages                 = {70},
  Publisher                = {American Physical Society},
  URL                      = {https://link.aps.org/doi/10.1103/RevModPhys.90.035005}
}

@Article{Pirandola2018,
  Title                    = {Advances in photonic quantum sensing},
  Author                   = {Pirandola, S. and Bardhan, B. R. and Gehring, T. and Weedbrook, C. and Lloyd, S.},
  Journal                  = {Nat. Photonics},
  Year                     = {2018},
  Number                   = {12},
  Pages                    = {724-733},
  Volume                   = {12},
  DOI                      = {10.1038/s41566-018-0301-6},
  ISSN                     = {1749-4893},
  URL                      = {https://doi.org/10.1038/s41566-018-0301-6}
}

@Article{Schnabel2016,
  author  = {R. Schnabel},
  title   = {Squeezed states of light and their applications in laser interferometers},
  journal = {Phys. Rep.},
  year    = {2017},
  volume  = {684},
  pages   = {1 - 51},
  issn    = {0370-1573},
  url     = {http://www.sciencedirect.com/science/article/pii/S0370157317300595},
}

@Article{sekatski2017quantum,
  author    = {Sekatski, Pavel and Skotiniotis, Michalis and Ko{\l}ody{\'n}ski, Janek and D{\"u}r, Wolfgang},
  title     = {Quantum metrology with full and fast quantum control},
  journal   = {Quantum},
  year      = {2017},
  volume    = {1},
  pages     = {27},
  publisher = {Verein zur F{\"o}rderung des Open Access Publizierens in den Quantenwissenschaften},
  url       = {https://quantum-journal.org/papers/q-2017-09-06-27/},
}

@Article{smirne2016ultimate,
  author    = {Smirne, Andrea and Ko{\l}ody{\'n}ski, Jan and Huelga, Susana F and Demkowicz-Dobrza{\'n}ski, Rafa{\l}},
  title     = {Ultimate precision limits for noisy frequency estimation},
  journal   = {Phys. Rev. Lett.},
  year      = {2016},
  volume    = {116},
  number    = {12},
  pages     = {120801},
  publisher = {APS},
  url       = {https://journals.aps.org/prl/pdf/10.1103/PhysRevLett.116.120801},
}

@Article{tsang2016quantum,
  author    = {Tsang, Mankei and Nair, Ranjith and Lu, Xiao-Ming},
  title     = {Quantum theory of superresolution for two incoherent optical point sources},
  journal   = {Phys. Rev. X},
  year      = {2016},
  volume    = {6},
  number    = {3},
  pages     = {031033},
  publisher = {APS},
  url       = {https://journals.aps.org/prx/pdf/10.1103/PhysRevX.6.031033},
}

@Article{unden2016quantum,
  author    = {Unden, Thomas and Balasubramanian, Priya and Louzon, Daniel and Vinkler, Yuval and Plenio, Martin B and Markham, Matthew and Twitchen, Daniel and Stacey, Alastair and Lovchinsky, Igor and Sushkov, Alexander O and others},
  title     = {Quantum metrology enhanced by repetitive quantum error correction},
  journal   = {Phys. Rev. Lett.},
  year      = {2016},
  volume    = {116},
  number    = {23},
  pages     = {230502},
  publisher = {APS},
  url       = {https://journals.aps.org/prl/pdf/10.1103/PhysRevLett.116.230502},
}

@Article{viola1999dynamical,
  author    = {Viola, Lorenza and Knill, Emanuel and Lloyd, Seth},
  title     = {Dynamical decoupling of open quantum systems},
  journal   = {Phys. Rev. Lett.},
  year      = {1999},
  volume    = {82},
  number    = {12},
  pages     = {2417},
  publisher = {APS},
  url       = {https://journals.aps.org/prl/pdf/10.1103/PhysRevLett.82.2417},
}

@Article{zhou2018achieving,
  author    = {Zhou, Sisi and Zhang, Mengzhen and Preskill, John and Jiang, Liang},
  title     = {Achieving the Heisenberg limit in quantum metrology using quantum error correction},
  journal   = {Nat. Commun.},
  year      = {2018},
  volume    = {9},
  number    = {1},
  pages     = {78},
  publisher = {Nature Publishing Group},
  url       = {https://www.nature.com/articles/s41467-017-02510-3},
}

\clearpage

\section*{end matter}

\subsection{Proof of Lemma~\ref{lem:core} and Theorem~\ref{thm:1}}
\label{app:proof}

This section is based on the methods applied in \cite{zhou2018achieving} for a similar problem with a different operator space.

\textbf{(1) Necessity.}
If $G=g_1\openone+\sum_\alpha g_\alpha A_\alpha$, then for any $\ket{\psi_0},\ket{\psi_1}$ satisfying $\forall_\alpha \braket{\psi_0|A_\alpha|\psi_0}=\braket{\psi_1|A_\alpha|\psi_1}$ automatically $\braket{\psi_1|G|\psi_1}-\braket{\psi_0|G|\psi_0}=0$, so any decoherence-free subspace is automatically insensitive to parameter changes.

\textbf{(2) Sufficiency.}
Let us define the inner product of two Hermitian matrices $A$, $B$ acting on $\mathcal H_S$ as the trace of their product, $\tr(AB)$. Then any $G$ satisfying
\begin{equation}
   G\notin \t{span}_{\mathbb R}\{\openone,A_\alpha\} 
\end{equation}
has a unique decomposition $G=G_\parallel+G_\perp$, where $G_\parallel\in\t{span}_{\mathbb R}\{\openone,A_\alpha\}$ and $\forall_\alpha\tr(G_\perp A_\alpha)=\tr(G_\perp)=0$. Next, $G_\perp$ may be always written as $G_\perp=\frac{1}{2}(\tr|G_\perp|)(\rho_1-\rho_0)$ where both $\rho_0,\rho_1$ are mutually orthogonal matrices with $\tr(\rho_{0/1})=1$ and $|G_\perp|:=\sqrt{G_\perp^2}$. Let $\ket{\psi_0},\ket{\psi_1}\in \mathcal H_S\otimes\mathcal H_A$ be the purifications of $\rho_0,\rho_1$ (i.e. $\forall_{i=0,1}\tr_A(\ket{\psi_i}\bra{\psi_i})=\rho_i$) with orthogonal support on the ancilla $\mathcal H_A$. From $\forall_\alpha\tr(G_\perp A_\alpha)=0$, such constructed states satisfy
\begin{equation}
    \forall_\alpha\braket{\psi_0|(A_\alpha\otimes \openone)|\psi_0}=\braket{\psi_1|(A_\alpha\otimes \openone)|\psi_1}\,.
\end{equation}
Since they have orthogonal support on the ancilla, they satisfy 
\begin{equation}
    \braket{\psi_0|(A_\alpha\otimes \openone)|\psi_1}=0.
\end{equation}
Finally, the sensitivity for the signal generated by $G$ is 
\begin{equation}
\begin{split}
    &\braket{\psi_1|(G\otimes \openone)|\psi_1}-\braket{\psi_0|(G\otimes \openone)|\psi_0}\\
    =&\braket{\psi_1|(G_\perp\otimes \openone)|\psi_1}-\braket{\psi_0|(G_\perp\otimes \openone)|\psi_0}\\
    =&\tr(G_\perp (\rho_1-\rho_0))=
    \tr\left(G_\perp \tfrac{G_\perp}{\frac{1}{2}\tr|G_\perp|}\right)\\
    =&\frac{2\tr(G_\perp^2)}{\tr |G_\perp|}.
    \end{split}
\end{equation}

\textbf{(3) Optimization as SDP.}
Let $\ket{\tilde \psi_0},\ket{\tilde \psi_1}\in\mathcal H_S\otimes \mathcal H_A$ be a pair of orthogonal states. We define:
\begin{equation}
    \tilde \rho_0=\tr_A(\ket{\tilde \psi_0}\bra{\tilde \psi_0}),\quad \tilde \rho_1=\tr_A(\ket{\tilde \psi_1}\bra{\tilde \psi_1}),
\end{equation}
and $\tilde G=\tilde\rho_1-\tilde\rho_0$. Since both states are normalized, $\tr|\tilde G|\leq2$ and $\tr(\tilde G)=0$.
Condition $\braket{\tilde\psi_0|(A_\alpha\otimes \openone)|\tilde\psi_0}=\braket{\tilde\psi_1|(A_\alpha\otimes \openone)|\tilde\psi_1}$ is equivalent to $\tr(A_\alpha \tilde G)=0$.

Note also that the support of $\ket{\tilde\psi_{0/1}}$ on $\mathcal H_A$ has no impact on sensitivity $\braket{\tilde\psi_1|(G\otimes \openone)|\tilde\psi_1}-\braket{\tilde\psi_0|(G\otimes \openone)|\tilde\psi_0}$. Moreover, if these supports are mutually orthogonal, the condition $\braket{\psi_0|A_\alpha\otimes \openone|\psi_1}=0$ is satisfied automatically.
Therefore, looking for an optimal choice of $\ket{\tilde\psi_0},\ket{\tilde\psi_1}$, without loss of generality we may restrict to the ones with orthogonal support on $\mathcal H_A$.

Therefore, the whole problem may be formulated as optimization over a hermitian matrix $\tilde G$ satisfying:

\begin{equation}
    \begin{split}
        &\max_{\tilde G}\tr(G \tilde G)\\
        &\tr(\tilde G)=0,\quad\forall_\alpha\tr(A_\alpha \tilde G)=0,\quad \tr|\tilde G|\leq 2.
    \end{split}
\end{equation}
The constraint $\tr|\tilde G|\leq 2$ may be written as linear constraints by introducing as an additional variable a hermitian matrix $X$:
\begin{equation}
    \tr|\tilde G|\leq 2 \Leftrightarrow \exists_{X\succeq 0}:-X\preceq\tilde G\preceq X,\tr(X)\leq 2,
\end{equation}
so finally we obtain SDP problem:
\begin{equation}
    \begin{split}
        &\max_{\tilde G,X}\tr(G \tilde G)\\
        &\tr(\tilde G)=0,\quad\forall_\alpha\tr(A_\alpha \tilde G)=0,\\
        &X+\tilde G\succeq0,\quad X-\tilde G\succeq0,\quad \tr(X)\leq2.
    \end{split}
\end{equation}
Alternatively, the problem may by formulated as a minimization of the Lagrange dual problem (also solvable by SDP)~\cite{zhou2018achieving}.

\subsection{Lamb shift Hamiltonian}
\label{app:lamb}
The Lamb shift Hamiltonian is given by
\begin{equation}
\label{eq:lamb}
  H_{LS}=\sum_\nu\sum_{\alpha,\beta}S_{\alpha\beta}(\nu){L_\alpha^\nu}^\dagger L_\beta^\nu,
\end{equation}
where $S_{\alpha\beta}(\nu)$ is the dispersive counterpart of
$\gamma_{\alpha\beta}(\nu)$, the two being related by a Kramers--Kronig
transform~\cite{breuer2002theory}. By construction, $H_{LS}$ always commutes
with $H_D$. Therefore, if the eigenstates $\ket{\psi_0},\ket{\psi_1}$ have
distinct eigenvalues $\lambda_0\neq\lambda_1$, $H_{LS}$ commutes also with
$G_{\t{eff}}$, and its only effect is an additional phase shift, which has to
be accounted for when constructing the estimator.

Note that once condition~\eqref{eq:relaxation} is imposed, the off-diagonal
elements $\braket{\psi_0|A_\alpha|\psi_1}$ are excluded explicitly rather than removed by the secular approximation, so decoupling from the noise no longer requires
a large splitting between $\lambda_0$ and $\lambda_1$. In that case the
effective generator is defined as $\Pi_{\mathcal C}G\Pi_{\mathcal C}$. Note,
however, that if $\ket{\psi_0},\ket{\psi_1}$ span a degenerate eigenspace of
$H_D$, then $H_{LS}$ may act nontrivially inside it, so that
$[H_{LS},G_{\t{eff}}]\neq 0$; this disturbs the estimation procedure by
introducing an additional rotation in a different
direction~\cite{de2013quantum}. It is therefore beneficial to keep
$|\lambda_1-\lambda_0|$ large compared with the Lamb shift.

Note also that the standard derivation of the master equation relies on
perturbation theory truncated at second order. 
Fourth-order corrections,  at least for noise arising from a classical fluctuating field~\cite{groszkowski2023simple},
result in fluctuations of the Lamb shift,
which leads to dephasing between $\ket{\psi_0}$ and $\ket{\psi_1}$ even when
$\Pi_{\mathcal C}A_\alpha\Pi_{\mathcal C}\propto\Pi_{\mathcal C}$ is satisfied.
Our discussion is therefore valid only in the regime where fourth-order
corrections are negligible.

\subsection{Proof of \autoref{thm:2}}
\label{app:proof2}

Condition $G\notin \vspan\{\openone,L_k,L_k^\dagger,L_k^\dagger L_j\}$ is equivalent to the existence of a two-dimensional code space $\mathcal C\subset \mathcal H_S\otimes \mathcal H_A$ satisfying the Knill-Laflamme conditions for the Lindblad operators \cite{zhou2018achieving}:
\begin{equation}
\label{eq:proofL}
\begin{split}
    \forall_{\alpha,\nu}\Pi_{\mathcal C}L_\alpha(\nu)\Pi_{\mathcal C}\propto\Pi_\mathcal C,\\
    \forall_{\alpha,\nu,\beta,\nu'}\Pi_{\mathcal C}L^\dagger_\alpha(\nu)L_\beta(\nu')\Pi_{\mathcal C}\propto\Pi_\mathcal C,
\end{split}
\end{equation}
with $\Pi_{\mathcal C}G\Pi_{\mathcal C}\not\propto\Pi_\mathcal C$. Note that the above conditions imply that the Lamb shift~\eqref{eq:lamb} acts
as a global phase on $\mathcal{C}$, $\Pi_{\mathcal{C}}H_{LS}\Pi_{\mathcal{C}}
\propto\Pi_{\mathcal{C}}$, so it need not be discussed further.

Similarly, the condition $G\notin \mathrm{span}_{\mathbb C}\{\openone,A_\alpha,A_\alpha A_\beta\}$ is equivalent to the existence of a code space $\mathcal C\subset \mathcal H_S\otimes \mathcal H_A$ such that
\begin{equation}
\label{eq:proofA}
\begin{split}
    \forall_{\alpha}\Pi_{\mathcal C}A_\alpha\Pi_{\mathcal C}\propto\Pi_\mathcal C,\\
    \forall_{\alpha,\beta}\Pi_{\mathcal C}A_\alpha A_\beta\Pi_{\mathcal C}\propto\Pi_\mathcal C,
\end{split}
\end{equation}
with $\Pi_{\mathcal C}G\Pi_{\mathcal C}\not\propto\Pi_\mathcal C$.

We show that, for a fixed $\mathcal C$, \eqref{eq:proofL} implies \eqref{eq:proofA}, and that \eqref{eq:proofA} implies the existence of a control Hamiltonian $H_C$ whose eigenspace is $\mathcal C$ and for which \eqref{eq:proofL} holds.

\textbf{(1) \eqref{eq:proofL} $\Rightarrow$ \eqref{eq:proofA}}

Recall that the Lindblad operator labeled by $(\alpha,\nu)$ is related to $A_\alpha$ via
\begin{equation}
 L_\alpha(\nu)\equiv\sum_{\epsilon'-\epsilon=\nu}\Pi_\epsilon A_\alpha \Pi_{\epsilon'}.
\end{equation}
Therefore, $A_\alpha$ decomposes as
\begin{equation}
A_\alpha=\sum_\nu L_\alpha(\nu).
\end{equation}
Summing \eqref{eq:proofL} over $\nu$ and $\nu'$ immediately gives \eqref{eq:proofA}.

\textbf{(2) \eqref{eq:proofA} $\Rightarrow$ $\exists\,H_C$ such that \eqref{eq:proofL} holds}

We choose $H_C$ such that the system Hamiltonian has only two eigenspaces, namely $\mathcal C$ and $\mathcal C^\perp$, separated by an energy gap $\nu_0$. The Lindblad operators are then
\begin{equation}
    \begin{split}
        L_\alpha(0)&=\Pi_\mathcal C A_\alpha\Pi_\mathcal C+\Pi_{\mathcal C^\perp} A_\alpha\Pi_{\mathcal C^\perp},\\
        L_\alpha(-\nu_0)&=\Pi_{\mathcal C^\perp} A_\alpha\Pi_{\mathcal C},\\
        L_\alpha(+\nu_0)&=\Pi_{\mathcal C} A_\alpha\Pi_{\mathcal C^\perp},
    \end{split}
\end{equation}
for all $\alpha$. By construction, $\Pi_{\mathcal C}L_\alpha(\pm\nu_0)\Pi_{\mathcal C}=0$, while $\Pi_{\mathcal C}A_\alpha\Pi_{\mathcal C}\propto\Pi_{\mathcal C}$ implies
$\Pi_{\mathcal C}L_\alpha(0)\Pi_{\mathcal C}\propto\Pi_{\mathcal C}$.
Next, after projection onto $\mathcal C$, the only potentially nonzero quadratic terms are
\[
\Pi_{\mathcal C}L^\dagger_\alpha(0)L_\beta(0)\Pi_{\mathcal C}
\qquad\text{and}\qquad
\Pi_{\mathcal C}L^\dagger_\alpha(-\nu_0)L_\beta(-\nu_0)\Pi_{\mathcal C}.
\]
For the first term, from $\Pi_{\mathcal C}A_\alpha\Pi_{\mathcal C}\propto\Pi_{\mathcal C}$ one has
\begin{equation}
\Pi_{\mathcal C}L^\dagger_\alpha(0)L_\beta(0)\Pi_{\mathcal C}
=
\Pi_{\mathcal C}A_\alpha\Pi_{\mathcal C}A_\beta\Pi_{\mathcal C}
\propto \Pi_{\mathcal C}.
\end{equation}
To analyze the second term, we expand the condition
$\Pi_{\mathcal C}A_\alpha A_\beta\Pi_{\mathcal C}\propto \Pi_{\mathcal C}$ as
\begin{equation}
\Pi_{\mathcal C}(A_\alpha\Pi_\mathcal C+A_\alpha\Pi_{\mathcal C^\perp})(\Pi_\mathcal C A_\beta+\Pi_{\mathcal C^\perp}A_\beta)\Pi_\mathcal C\propto \Pi_{\mathcal C}.
\end{equation}
Using $\Pi_\mathcal C\Pi_{\mathcal C^\perp}=0$ and $\Pi_{\mathcal C}A_\alpha\Pi_{\mathcal C}\propto\Pi_{\mathcal C}$, this implies
\begin{equation}
\Pi_\mathcal C A_\alpha \Pi_{\mathcal C^\perp}A_\beta\Pi_\mathcal C\propto \Pi_\mathcal C.
\end{equation}
which means $\Pi_\mathcal C L^\dagger_\alpha(-\nu_0)L_\beta(-\nu_0)\Pi_\mathcal C\propto\Pi_{\mathcal C}$.
This proves \eqref{eq:proofL} for the chosen $H_C$.

\subsection{No two-dimensional trivial subspace for spin-1}
\label{app:example}

Let $\{S_i\}_{i\in\{x,y,z\}}$ be the irreducible spin-$1$ representation on a three-dimensional Hilbert space $\mathcal H_S$,  given in the standard
basis by
\begin{equation}
\begin{aligned}
   S_x&=\frac{1}{\sqrt{2}}\begin{bmatrix}0&1&0\\1&0&1\\0&1&0\end{bmatrix}, &
   S_y&=\frac{1}{\sqrt{2}}\begin{bmatrix}0&-i&0\\i&0&-i\\0&i&0\end{bmatrix}, \\[4pt]
   S_z&=\begin{bmatrix}1&0&0\\0&0&0\\0&0&-1\end{bmatrix}. &&
\end{aligned}
\end{equation}
Then there exists no two-dimensional subspace $\mathcal C\subset \mathcal H_S$ satisfying
\begin{equation}
\label{eq:prep}
    \forall_{i\in\{x,y,z\}}\ \Pi_\mathcal C S_i \Pi_\mathcal C=\lambda_i \Pi_\mathcal C.
\end{equation}

\begin{proof}

We proceed by contradiction. We use the fact that each $S_i$ has eigenvalues $-1,0,+1$, as well as the commutation relations.

First, we show that for any direction $i\in\{x,y,z\}$ the equality $\Pi_\mathcal C S_i \Pi_\mathcal C=\lambda_i \Pi_\mathcal C$ implies $\lambda_i=0$. We start with the $z$ direction. Since any two two-dimensional subspaces of the three-dimensional space $\mathcal H_S$ have a non-empty intersection, we may choose a vector $\ket{u}\in\mathcal C\cap\vspan\{\ket{0},\ket{+1}\}$. Since $\ket{u}\in\mathcal C$, we have $\braket{u|S_z|u}=\lambda_z$, and since $\ket{u}\in\vspan\{\ket{0},\ket{+1}\}$, it follows that $\braket{u|S_z|u}\geq 0$, which together implies $\lambda_z\geq 0$. Similarly, there exists another vector $\ket{v}\in\mathcal C\cap\vspan\{\ket{0},\ket{-1}\}$, which implies $\lambda_z\leq 0$. Therefore $\lambda_z=0$. The same reasoning applies for $x$ and $y$.

Since $S_i$ is traceless, $\Pi_{\mathcal C} S_i \Pi_{\mathcal C}=0$ implies also $\Pi_{\mathcal C^\perp} S_i \Pi_{\mathcal C^\perp}=0$ (where $\Pi_{\mathcal C^\perp}$ denotes the projection onto the subspace orthogonal to $\mathcal C$). Therefore, \eqref{eq:prep} implies
\begin{equation}
\label{eq:xyzp}
\forall_{i\in\{x,y,z\}}\quad
S_i=\Pi_\mathcal C S_i \Pi_{\mathcal C^\perp}+\Pi_{\mathcal C^\perp} S_i \Pi_{\mathcal C}.
\end{equation}
Since $\Pi_\mathcal C\Pi_{\mathcal C^\perp}=0$, this leads to
\begin{equation}
\label{eq:comm}
[S_x,S_y]=\Pi_\mathcal C [S_x,S_y]\Pi_\mathcal C +\Pi_{\mathcal C^\perp}[S_x,S_y]\Pi_{\mathcal C^\perp}.
\end{equation}
Since $S_{z}=-i[S_{x},S_{y}]$, it must simultaneously have the forms \eqref{eq:xyzp} and \eqref{eq:comm}, which implies $S_{z}\equiv 0$. This contradicts the fact that $S_{z}$ has eigenvalues $-1,0,+1$. Therefore, there exists no two-diemnsional subspace $\mathcal C\subset \mathcal H_S$ satisfying \eqref{eq:prep}.

\end{proof}

\setcounter{section}{0}
\setcounter{equation}{0}
\renewcommand{\theequation}{S\arabic{equation}}

\clearpage
\onecolumngrid

\begin{center}
  {\large\bfseries Supplemental Material to:\\Protecting Heisenberg scaling in quantum metrology via engineered dressed states}\\[4mm]
  Wojciech G\'orecki and Christiane P. Koch\\[3mm]
  {\itshape\footnotesize
  Freie Universit\"at Berlin, Fachbereich Physik and Dahlem Center for
  Complex Quantum Systems, Arnimallee 14, 14195 Berlin, Germany}
\end{center}
\vspace{4mm}

\section{Error bound for the projected dynamics}
\label{app:zeno}

The bound \eqref{eq:timesscales} on the difference between the exact and the
projected evolution is a special case of a more general theorem derived in
Ref.~\cite{Burgarth2022oneboundtorulethem}, which treats the situation of
both the strong drive and the perturbation being time-dependent. For clarity, we
quote it here in the original notation of Ref.~\cite{Burgarth2022oneboundtorulethem} and indicate which terms drop out in our setting. For
$U_\kappa(t)=\mathcal T\exp\big(-i\int_0^t\!ds\,[\kappa H_0(s)+H_1(s)]\big)$
with $H_0(s)=\sum_{\ell=1}^{m}E_\ell(s)P_\ell(s)$, where $E_\ell(s)$ and
$P_\ell(s)$ are the instantaneous eigenvalues and eigenprojectors of the driving
Hamiltonian, Ref.~\cite[Eq.~(4.43)]{Burgarth2022oneboundtorulethem} proves, for
$t\in[0,T]$,
\begin{align}
\label{eq:burgarth}
\Big\|U_\kappa(t)&-\mathcal T e^{-i\int_0^t ds\,[\kappa H_0(s)+H_Z(s)+A(s)]}\Big\|
\le\frac{\sqrt m}{\kappa\eta}\big(1+T\|A\|_{\infty,T}+2T\|H_1\|_{\infty,T}\big)
\nonumber\\
&\times\Big[\Big(2+\tfrac{\eta'}{\eta}T\Big)
\big(\|A\|_{\infty,T}+\|H_1\|_{\infty,T}\big)
+T\big(\|\dot A\|_{\infty,T}+\|\dot H_1\|_{\infty,T}
+2\|A\|_{\infty,T}\|H_1\|_{\infty,T}\big)\Big],
\end{align}
with $\|X\|_{\infty,T}=\sup_{s\in[0,T]}\|X(s)\|$. Here
$H_Z(s)=\sum_\ell P_\ell(s)H_1(s)P_\ell(s)$ is the Zeno Hamiltonian, $m$ is the
number of distinct eigenvalues of $H_0$, and $\eta$ is its minimal spectral gap.
The operator $A$ and the quantity $\eta'$ are defined in
Ref.~\cite[Eqs.~(4.14) and (4.15)]{Burgarth2022oneboundtorulethem}; both are
built from the time derivatives of the spectral projections and therefore
vanish whenever $H_0$ is constant.

This is the case here: $\kappa H_0=H_D$, so
that its spectral projections are constant, $P_\ell=\Pi_\epsilon$, and hence
$A=0$ and $\eta'=0$. In our notation $\kappa\eta=\nu_{\min}$. With
$H_1(t)=V(t)=\delta\omega\,G+H_{\t{cl}}(t)$, the Zeno Hamiltonian reproduces
exactly \eqref{eq:hamproj}, $H_D+H_Z(t)=H^{\t{eff}}(t)$, and
\eqref{eq:burgarth} reduces to the bound quoted in the main text, which we
repeat here for convenience:
\begin{equation}
\label{eq:timesscales2}
\Big\|\mathcal{T}e^{-i\int_0^{T}\!dt\,H(t)}
-\mathcal{T}e^{-i\int_0^{T}\!dt\,H_{\t{eff}}(t)}\Big\|
\le
\frac{\sqrt m}{\nu_{\min}}\big(1+2T\|V\|_{\infty,T}\big)
\big(2\|V\|_{\infty,T}+T\|\dot V\|_{\infty,T}\big).
\end{equation}
Expanding the product gives four contributions,
\begin{equation}
\label{eq:fourterms}
\frac{\sqrt m}{\nu_{\min}}
\Big[\,2\|V\|_{\infty,T}+T\|\dot V\|_{\infty,T}
+4T\|V\|_{\infty,T}^{2}+2T^{2}\|V\|_{\infty,T}\|\dot V\|_{\infty,T}\Big].
\end{equation}
The first is time-independent and yields the validity condition
$\sqrt m\,\|V\|_{\infty,T}\ll\nu_{\min}$; the remaining three restrict the
interrogation time to
\begin{equation}
\label{eq:Tconstraints}
T\ll\min\left(\frac{\nu_{\min}}{\sqrt m\,\|\dot V\|_{\infty,T}},\;
\frac{\nu_{\min}}{\sqrt m\,\|V\|_{\infty,T}^{2}},\;
\sqrt{\frac{\nu_{\min}}{\sqrt m\,\|V\|_{\infty,T}\|\dot V\|_{\infty,T}}}\right).
\end{equation}

Consider now $N$ identical subsystems. The perturbation is then a sum of
operators acting on the individual subsystems, so both $\|V\|_{\infty,T}$ and
$\|\dot V\|_{\infty,T}$ grow linearly with $N$. At the same time, the number $m$
of distinct eigenvalues of the total dressing Hamiltonian grows with $N$ as
well: linearly if the eigenvalues of the single-subsystem $H_D$ are
commensurate, and polynomially for a generic spectrum. Requiring that
\eqref{eq:Tconstraints} still holds at a fixed total interrogation time, the
second and third constraints both give $\nu_{\min}\gtrsim\sqrt m\,N^{2}$, so
that $\nu_{\min}$ must grow at least as $N^{5/2}$, and faster for a generic
spectrum. The bound \eqref{eq:timesscales2} is presumably not tight in this
regime, and the actual scaling of $\nu_{\min}$ needed to keep the decoupling
effective would have to be established for a concrete model; some growth with
$N$ is nevertheless to be expected.

\section{White noise as an infinite-temperature bath: an explicit application of Theorem~2}
\label{app:perpendicular}

This section serves two purposes. First, it shows that purely white noise has to
be treated as an infinite-temperature bath, so \autoref{thm:2}, rather than \autoref{thm:1},
is the relevant one. Second, it makes the content of \autoref{thm:2}
concrete by constructing the required control Hamiltonian explicitly, and
compares the resulting protocol with two simpler ones.

We consider a qubit with
$H^{\t{free}}_S=\omega\sigma_z=(\omega_0+\delta\omega)\sigma_z$, so that
$G=\sigma_z$, coupled to the environment through a perpendicular field,
\begin{equation}
    H_I=\sigma_x B_x,\qquad \braket{B_x(t)B_x(t')}=\gamma\,\delta(t-t').
\end{equation}

\subsection{Two master equations}

Since the noise is $\delta$-correlated, its correlation time vanishes and the
averaged evolution is given, on any time scale, by the singular-coupling master
equation~\cite{breuer2002theory} with the single Lindblad operator $L=\sigma_x$,
\begin{equation}
\label{eq:masterperpexact}
    \frac{d\rho}{dt}=-i[H_S,\rho]+\gamma\left(\sigma_x\rho\sigma_x-\rho\right).
\end{equation}
This is the canonical example in which HNLS permits HS in the presence of
non-trivial noise~\cite{zhou2018achieving}: since $\sigma_x^2=\openone$, the
Lindblad span reduces to $\vspan_{\mathbb C}\{\openone,\sigma_x\}$, and
$G=\sigma_z\notin\vspan_{\mathbb C}\{\openone,\sigma_x\}$.

If the coupling is weak compared to the qubit frequency, $\gamma\ll\omega$,
then for coarse-graining times $\tau_{\t{cg}}\gg 1/\omega$ one may instead derive
the master equation in the weak-coupling limit~\cite{breuer2002theory}. Projecting $\sigma_x$
onto the eigenspaces of $H_S$ leaves two jump operators,
\begin{equation}
\label{eq:Lbare}
    L_{+}=\ket{0}\bra{0}\sigma_x\ket{1}\bra{1}=\ket{0}\bra{1},\qquad
    L_{-}=\ket{1}\bra{1}\sigma_x\ket{0}\bra{0}=\ket{1}\bra{0},
\end{equation}
where $\ket{0},\ket{1}$ are the eigenstates of $\sigma_z$, so that $\ket{0}$ is
the higher-energy state and $L_{+}$ describes excitation. Since the noise
is white, the bath spectral density takes the same value at every frequency, in
particular $\gamma(2\omega)=\gamma(-2\omega)=\gamma$, which is precisely the
infinite-temperature case. The master equation then reads
\begin{equation}
\label{eq:masterperpapp}
    \frac{d\rho'}{dt}=-i[H_S,\rho']
    +\gamma\sum_{i\in\{+,-\}}\Big(L_i\rho'{L_i}^\dagger
    -\tfrac{1}{2}\{{L_i}^\dagger L_i,\rho'\}\Big).
\end{equation}
Here HNLS fails, since
$\sigma_z=L_-^\dagger L_--L_+^\dagger L_+$ lies in the Lindblad span, so HS cannot be recovered by QEC applied to \eqref{eq:masterperpapp}.

The difference admits the following interpretation. The noise generated by
$\sigma_x$ can be viewed as a bit flip acting at a random instant. By itself, a
bit flip preserves the relative phase. However, since the system rotates
rapidly about $\sigma_z$, the two terms acquire different phases depending on
the exact time at which the flip occurs. The phase information is thereby lost,
and the action of such a random bit flip becomes indistinguishable from
independent excitation and relaxation of equal strength.

This picture is confirmed by the exact solution of the free evolution, i.e., in the absence of QEC.
Solving
\eqref{eq:masterperpexact} analytically gives
\begin{equation}
\label{eq:exact}
\rho(t)=\frac{1}{2}
\begin{pmatrix}
1+z_0e^{-2\gamma t} & e^{-\gamma t}\big(X-iY\big)\\[4pt]
e^{-\gamma t}\big(X+iY\big) & 1-z_0e^{-2\gamma t}
\end{pmatrix},
\end{equation}
with
\begin{align}
X&=\Big(\cos\Omega t+\frac{\gamma}{\Omega}\sin\Omega t\Big)x_0
   -\frac{2\omega}{\Omega}\sin(\Omega t)\,y_0,\\
Y&=\frac{2\omega}{\Omega}\sin(\Omega t)\,x_0
   +\Big(\cos\Omega t-\frac{\gamma}{\Omega}\sin\Omega t\Big)y_0,
\end{align}
and $\Omega=\sqrt{4\omega^2-\gamma^2}$, whereas \eqref{eq:masterperpapp} yields
\begin{equation}
\label{eq:simply}
\rho'(t)=\frac{1}{2}
\begin{pmatrix}
1+z_0e^{-2\gamma t} & (x_0-iy_0)\,e^{-\gamma t}e^{-2i\omega t}\\[4pt]
(x_0+iy_0)\,e^{-\gamma t}e^{+2i\omega t} & 1-z_0e^{-2\gamma t}
\end{pmatrix}.
\end{equation}
The solution \eqref{eq:simply} reproduces \eqref{eq:exact} up to relative
corrections of order $\mathcal O(\gamma/\omega)$, so on timescales
$t\gg 1/\omega$ the two agree.

From the point of view of the underlying microscopic model, HS can therefore be achieved without any dressing, provided QEC is applied on timescales shorter than $1/\omega$, i.e. before \eqref{eq:masterperpexact} averages to \eqref{eq:masterperpapp}.

\subsection{Recovering Heisenberg scaling in the weak-coupling limit: Theorem~2}

We now show that a suitable control Hamiltonian restores HS for this model also
in the weak-coupling limit.
As the bath is effectively at infinite temperature, \autoref{thm:2} is the relevant criterion. Its condition is
satisfied: with the single coupling operator $A=\sigma_x$ one has
$A^2=\openone$, so that
$\vspan_{\mathbb C}\{\openone,A,A^2\}=\vspan_{\mathbb C}\{\openone,\sigma_x\}$
and $G=\sigma_z$ lies outside it. A control Hamiltonian generating a correctable
subspace must therefore exist. Following the construction in End
Matter~\ref{app:proof2}, we append a qubit ancilla and take
\begin{equation}
\label{eq:Hcperp}
    H_C=-\omega_0\,\sigma_z\otimes\openone_A-\lambda\,\Pi_{\mathcal C},
    \qquad
    \Pi_{\mathcal C}=\ket{00}\bra{00}+\ket{11}\bra{11},
\end{equation}
with $\lambda>0$, so that $H_D=-\lambda\Pi_{\mathcal C}$ has exactly two
eigenspaces: the code space $\mathcal C=\vspan\{\ket{00},\ket{11}\}$ at energy
$-\lambda$, and $\mathcal C^\perp=\vspan\{\ket{01},\ket{10}\}$ at energy $0$.
Since $\sigma_x\otimes\openone_A$ maps each of these subspaces onto the other,
the dephasing component vanishes, $L_{0}=0$, and the only jump operators are
\begin{equation}
\label{eq:Lperp}
\begin{split}
    L_{-}&=\Pi_{\mathcal C}\,(\sigma_x\otimes\openone_A)\,\Pi_{\mathcal C^\perp}
    =\ket{00}\bra{10}+\ket{11}\bra{01},\\
    L_{+}&=\Pi_{\mathcal C^\perp}\,(\sigma_x\otimes\openone_A)\,\Pi_{\mathcal C}
    =\ket{10}\bra{00}+\ket{01}\bra{11},
\end{split}
\end{equation}
both occurring at the same rate $\gamma$. The subspace $\mathcal C$ is not
decoherence-free, but correctable. The signal survives, since
\begin{equation}
    \Pi_{\mathcal C}\,(G\otimes\openone_A)\,\Pi_{\mathcal C}
    =\ket{00}\bra{00}-\ket{11}\bra{11}\not\propto\Pi_{\mathcal C},
\end{equation}
and HS is therefore recovered by QEC applied to the averaged evolution.

The construction above serves 
to illustrate the usefulness of \autoref{thm:2}: it
provides a control Hamiltonian for which QEC restores HS \emph{after} the
evolution has been averaged to a weak-coupling master equation, that is, one
whose Lindblad operators act between eigenstates of the system Hamiltonian. While two
simpler protocols reach the same result, a comparison is instructive:
Without any dressing, HS remains achievable provided QEC is applied at intervals
much shorter than $1/\omega$, so that the dynamics is described by
\eqref{eq:masterperpexact} rather than \eqref{eq:masterperpapp} and the standard
procedure of Ref.~\cite{zhou2018achieving} applies. Alternatively, the control much simpler than \eqref{eq:Hcperp}, namely $H_C=-\omega_0\,\sigma_z$, leaves $H_D=0$;
all transition frequencies then collapse to a scale set by the unknown
$\delta\omega$, the secular decomposition no longer separates
$\ket{0}\bra{1}$ from $\ket{1}\bra{0}$, and \eqref{eq:masterperpexact} remains
the appropriate description over a correspondingly longer window, so that QEC
may be applied at a considerably lower rate.
In all three cases, the mechanism enabling HS is the same: preventing the
coupling operators from splitting into Bohr-frequency components, which is what
generates the jump operators that destroy the code. This can be achieved either
by more frequent QEC~\cite{zhou2018achieving} or by proper dressing, as we
suggest here. The control \eqref{eq:Hcperp} makes the code space an isolated
eigenspace of $H_D$, so that the resulting components act within it, while
$H_C=-\omega_0\sigma_z$ removes the level splitting altogether, leaving a single
Bohr frequency and hence no decomposition at all.

\end{document}